\documentclass[runningheads]{llncs}        

\makeatletter
\newcommand{\@chapapp}{\relax}%
\makeatother

\usepackage{scalerel}
\usepackage{orcidlink}

\usepackage{amssymb}
\usepackage{amsmath}

\usepackage{amsthm}
\usepackage{proof}
\usepackage{stmaryrd}
\usepackage{placeins}

\usepackage{booktabs}
\usepackage{dsfont}

\usepackage{bbm}
\usepackage{graphicx}
\usepackage{listings}
\usepackage{bussproofs}
\usepackage{floatflt} 

\usepackage{ifthen}
\usepackage{xspace}
\usepackage{xcolor}
\usepackage{pdfpages}
\usepackage{xspace}
\usepackage{url}

\usepackage[dash,dot]{dashundergaps}
\DeclareOldFontCommand{\rm}{\normalfont\rmfamily}{\mathrm}
\DeclareOldFontCommand{\sf}{\normalfont\sffamily}{\mathsf}
\DeclareOldFontCommand{\tt}{\normalfont\ttfamily}{\mathtt}
\DeclareOldFontCommand{\bf}{\normalfont\bfseries}{\mathbf}
\DeclareOldFontCommand{\it}{\normalfont\itshape}{\mathit}
\DeclareOldFontCommand{\sl}{\normalfont\slshape}{\@nomath\sl}
\DeclareOldFontCommand{\sc}{\normalfont\scshape}{\@nomath\sc}
\DeclareRobustCommand*\cal{\@fontswitch\relax\mathcal}
\DeclareRobustCommand*\mit{\@fontswitch\relax\mathnormal}

\colorlet{keywordcolor}{blue!50!black}
\colorlet{commentcolor}{green!60!black}
\colorlet{typecolor}{violet}
\newcommand{\sourcefont}{\ttfamily\small}
\newcommand{\commentfont}{\slshape\rmfamily\color{commentcolor}}

\lstdefinelanguage{ABS}{
        keywords={physical,duration,diff,differential,do,assert,this,new,data,type,def,case,of,local,class,interface,
        extends,implements,if,then,else,await,get,Fut,return,skip,while,module,
        import,export,from,to,suspend,delta,adds,modifies,removes,original,productline,
        features,core,corefeatures,optionalfeatures,after,when,product,hasAttribute,
        hasMethod,hasField,hasInterface,uses,root,extension,group,allof,oneof,require,
        stateupdate,object,main,objectupdate,classupdate,fi,
        exclude,original,ifin,ifout,opt,null,
        newgroup,data,thiscomp,in,joins,leaves,subtypeOf,wait,acquire,except,as,component,Pre,Abs
        },
        keywordstyle=\color{keywordcolor}\bfseries\sffamily,
        morekeywords=[2]{Unit, Int, Bool, Rat, List, Set, Pair, Fut, Maybe, String, Triple, Either, Map, Real},
        keywordstyle=[2]\color{typecolor},
        sensitive=true,
        comment=[l]{//},
        morecomment=[s]{/*}{*/},
        morestring=[b]",
        mathescape=true,
}

\lstdefinelanguage[v9]{Java}[]{Java}{
        morekeywords={module,requires,provides,uses,with,to,exports}
}
\lstdefinelanguage[ContextJ]{Java}[]{Java}{
        morekeywords={layer,with,without,proceed,before,after}
}
\lstdefinelanguage[FOP]{Java}[]{Java}{
        morekeywords={refines,original,Super}
}
\lstdefinelanguage[JastAdd]{Java}[]{Java}{
        morekeywords={aspect,syn,inh,lazy}
}

\lstdefinestyle{code}{
        basicstyle=\sourcefont\upshape,
        keywordstyle=\color{keywordcolor}\bfseries\sffamily,
        commentstyle=\commentfont,
        columns=fullflexible,
        mathescape=true,
        escapechar={\#},
        keepspaces=true,
        showstringspaces=false,
        aboveskip=8pt, 
        numbers=left,
        stepnumber=1, 
        numberstyle=\ttfamily\scriptsize\color{gray},
        numbersep=4pt,
        xleftmargin=1.5em,
        xrightmargin=1.5em,
        framexleftmargin=1.2em,
        framexrightmargin=1em,
        framextopmargin=0.5ex,
        breaklines=true,
        breakindent=3pt,
}

\lstdefinestyle{abs}{
        style=code,
        language=ABS,
}
\lstdefinestyle{java}{
        style=code,
            language=Java
}
\lstdefinestyle{java9}{
        style=code,
            language=[v9]Java
}
\lstdefinestyle{aspectj}{
        style=code,
        language=[AspectJ]Java
}
\lstdefinestyle{jastadd}{
        style=code,
        language=[JastAdd]Java
}
\lstdefinestyle{contextj}{
        style=code,
        language=[ContextJ]Java
}
\lstdefinestyle{FOP}{
        style=code,
        language=[FOP]Java
}
\lstdefinestyle{scala}{
        style=code,
        language=Scala,
        morekeywords={self}
}

\newcommand{\code}[2][]{\lstinline[style=code,basicstyle=\ttfamily\upshape,#1]|#2|}

\newcommand{\abs}[2][]{\code[style=abs,#1]{#2}}

\lstnewenvironment{srccode}[1][]{
        \minipage{\linewidth}
        \lstset{style=code,
        backgroundcolor=\color{white},
        rulecolor=\color{gray!50},
        frame=tblr,
        captionpos=b,
        #1}
}{\endminipage}

\lstnewenvironment{abscode}[1][]{
        \minipage{1\linewidth}
        \lstset{style=abs,
        backgroundcolor=\color{white},
        rulecolor=\color{gray!50},
        frame=tblr,
        captionpos=b,
        #1}
}{\endminipage}

\lstnewenvironment{javacode}[1][]{
        \minipage{\linewidth}
        \lstset{style=java,
        backgroundcolor=\color{gray!5},
        rulecolor=\color{gray!50},
        frame=tblr,
        captionpos=b,
        #1}
}{\endminipage}

\lstnewenvironment{jastaddcode}[1][]{
        \minipage{\linewidth}
        \lstset{style=jastadd,
        backgroundcolor=\color{gray!5},
        rulecolor=\color{gray!50},
        frame=tblr,
        captionpos=b,
        #1}
}{\endminipage}

 \makeatletter
\newcommand\BeraMonottfamily{%
  \def\fvm@Scale{0.85}
  \fontfamily{fvm}\selectfont
}
\makeatother

\lstdefinelanguage{KeYmaeraX}{%
  keywords={if,then,else,R,B,HP,Functions,ProgramVariables,Problem,End,Definitions,ArchiveEntry,Tactic,SharedDefinitions},%
  sensitive=true,
  morecomment=[s]{/*}{*/},
  deletestring=[d]',
  showstringspaces=false,
  commentstyle=\color{green},
  mathescape,
  escapeinside={/*@}{@*/}}[keywords]

\lstdefinelanguage{Bellerophon}{%
  language={},
  keywords={'R,'L,'_},%
  otherkeywords={;,<,|},
  sensitive=true,
  morecomment=[l]{//},
  morecomment=[s]{/*}{*/},
  morestring=[b]",
  deletestring=[d]',
  morestring=[d]`,
  showstringspaces=false,
  commentstyle=\fontseries{lc}\color{green}}[keywords]
  
\newcommand{\keycode}[1]{
    { \lstset{language=KeYmaeraX}
    \begin{lstlisting}
     #1
    \end{lstlisting}
    }
}

\usepackage{marginnote}

\newcommand{\EK}[1]{{#1}}

\newcommand{\COMMENT}[1]{}

\newcommand{\HABS}{\ensuremath{\mathtt{HABS}}\xspace}

\DeclareMathOperator*{\argmin}{\mathbf{argmin}}

\let\temp\phi
\let\phi\varphi
\let\varphi\temp

\makeatletter
\newcommand{\xRightarrow}[2][]{\ext@arrow 0359\Rightarrowfill@{#1}{#2}}

\newcommand{\cased}[1]{\ensuremath{
\left\{\begin{array}{ll}
#1
\end{array}\right.
}\xspace}

\newcommand{\synsep}{  \ | \ }
\newcommand{\xabs}[1]{\text{\abs{#1}}}

\newcommand{\sem}[1]{\ensuremath{ \llbracket #1 \rrbracket \xspace}}

\newcommand{\many}[1]{\overline{#1}}

\newcommand{\rulename}[1]{\textbf{\scriptsize(\textsf{#1})}}

\newcommand{\ddl}{\ensuremath{d\mathcal{L}}\xspace}
\newcommand{\dL}{\ddl}

\newcommand{\trace}{\ensuremath{\theta}\xspace}

\newcommand{\methodname}{\xabs{m}\xspace}
\newcommand{\classname}{\xabs{C}\xspace}
\newcommand{\statement}{\xabs{s}\xspace}

\usepackage{hyperref}

\usepackage[full]{mlrules}

\newcommand{\fullpaper}[1]{#1}

\newcommand{\coreid}{{\tt ceid}}

\newtheorem{notation}{Notation}

\newcommand{\CF}{\text{\tt CF}}
\newcommand{\cfexps}{\text{\tt E}}
\newcommand{\cfminus}{-}
\newcommand{\cfisinfty}{\text{\tt is\_infty}}
\newcommand{\cfispositive}{\text{\tt is\_positive}}

\newcommand{\tacontext}{\text{\tt TC}}
\newcommand{\TA}{\text{\tt TA}}
\newcommand{\tapreds}{\mathcal{C}}
\newcommand{\tamap}{\mathcal{T}}
\newcommand{\TI}{\text{\tt TI}}

\newcommand{\allmethodnames}{\vect{\xabs{C.m}}}
\newcommand{\allstatements}{\vect{\xabs{s}}}
\newcommand{\allrhs}{\vect{\xabs{rhs}}}

\newcommand{\futuremethod}{\mathit{M}}

\newcommand{\getanncontrolled}{\mathit{treq}}
\newcommand{\getanncontrolling}{\mathit{tctrl}}

\title{Type-Based Verification of Delegated Control in Hybrid~Systems (Full Version)}
\author{Eduard Kamburjan\,\orcidlink{0000-0002-0996-2543} \and Michael Lienhardt\,\orcidlink{0009-0009-9635-5757}}
\institute{
University of Oslo, Norway, \email{eduard@ifi.uio.no} \\
ONERA, Palaiseau, France, \email{michael.lienhardt@onera.fr} 
}

\titlerunning{Type-Based Verification of Delegated Control in Hybrid~Systems}

\begin{document}
\maketitle
\begin{abstract}
We present a post-region-based verification system for distributed hybrid systems modeled with Hybrid Active Objects. The post-region of a \EK{class} method is the region \EK{of the state space} where a physical process must be proven safe to ensure some object invariant.
Prior systems computed the post-region locally to a single object and could only verify systems where each object ensures its own safety, or relied on specific, non-modular communication patterns. The system presented here uses a \emph{type-and-effect system} to structure the interactions between objects and computes post-regions globally, but verifies them locally. Furthermore, we are able to handle systems with \emph{delegated} control: the object and method that shape the post-region change over time. We exemplify our approach with a model of a cloud-based hybrid system.
\end{abstract}

\section{Introduction}
\label{sec:intro}
Cyber-physical systems are notoriously difficult to design, maintain and analyze, and major innovation drivers such as the Internet-of-Things or Digital Twins pose additional challenges for formal modeling and verification.
For one, such systems are inherently distributed. For another, the controlling software may, contrary to classical control, use \emph{delegation} for the controlled process: parts of the controlling software may run on a cloud infrastructure which may restart the controlling processes, as well as reallocate to a different instance. Thus, the obligation for part of the control can be delegated to another instance. Formal guarantees are of critical importance, yet distributed hybrid models and delegation remain an open theoretical challenge. 

In this work, we present a system for modular deductive verification of distributed hybrid systems, 
which is able to handle delegated control.
Our approach is based on \emph{hybrid programs}: programs that contain constructs to express continuous evolution of their state.
Programming languages-based approaches for modeling of hybrid systems have recently gained increased research attention~\cite{Kamburjan21a,DBLP:conf/ictac/0001NP20,hybreb} and aim to provide a theory for hybrid systems that combines simulation, verification and usability. One of their advantages over low-level formalisms, such as hybrid automata~\cite{Rajeev93} or process algebras~\cite{CuijpersR05}, is the rich theory of modularity and structure available for programming languages that allows one to capture and analyze the adaptive structure of modern distributed cyber-physical systems. We show that hybrid programs can indeed provide the necessary structure to handle loose coupling and delegation, by integrating verification with a \emph{type-and-effect} system~\cite{185530}, a lightweight analysis technique for programs to keep track of side-effects in computational units.

\paragraph{Hybrid Active Objects.}
We use the Hybrid Active Object (HAO) concurrency model, which is one hybrid programming paradigm for distributed systems and extends Active Objects~\cite{BoerSHHRDJSKFY17} and is implemented in the \emph{Hybrid Abstract Behavioral Specification} (\texttt{HABS}) language~\cite{arxiv}.
%

A Hybrid Active Object $o$ is an Active Object that additionally encapsulates a physical process.
Only the discrete processes of $o$ may interact with the physical process of $o$.
A discrete process reacts on changes in the physical process using its suspension guards.
While the discrete process is active, the fields of the physical process can be accessed as normal fields, but when time advances such fields change their value according to the physical dynamics.

Previous work~\cite{Kamburjan21a,arxiv} introduced two verification systems for Hybrid Active Objects that verify object invariants:
Kamburjan~\cite{Kamburjan21a} gives a generalization of post-condition reasoning for object-oriented languages~\cite{key} to hybrid systems.
However, the system has one major drawback: it computes post-regions based on single classes -- it cannot handle interactions between multiple objects beyond checking conditions on the passed parameters. Thus, it is not able to use global information about the overall system to aid verification. 
On the other hand, an alternative system~\cite{arxiv} can handle more complex interactions, but suffers from a lack of modularity:
To make use of global structure, it only uses a small, rigidly defined syntactic subset of \texttt{HABS}. In particular, the structure of the overall system may not change and delegation is not possible.

\paragraph{Type-Based Deductive Verification}
In this work we present a novel verification system for HABS that goes beyond previous systems for object invariants: we use a type-and-effect system to enable \emph{post-regions} to be computed based on interactions between multiple objects. 
By using \EK{the structure provided by the type system}, we are able to integrate deductive verification systems with the modeling and analysis of cloud systems.

Given an object invariant $I$, we verify for each method \texttt{m} that when it suspends, $I$ holds until the next process runs. 
The \emph{post-region} $\mathit{pr}$ of a method is \EK{the part of the state space} where the dynamics must satisfy $I$ for this property to hold. For example, if another method \texttt{mctrl} executes every $n$ time units, then the post-region of \texttt{m} can be restricted to the next $n$ time units. It must, however, be ensured that \texttt{mctrl} is indeed called as specified. We say that \texttt{mctrl} is controlled if the global structure indeed ensures that it is called every $n$ time units. The caller of a controlled method is a \emph{controlling} method.

Our system can verify \EK{\emph{delegated}} control, where the controlling method changes during the lifetime of a controlled object.  Consider again the method \texttt{mctrl} from above. 
It can be called every $n$ time units \emph{by another object} and the object calling \texttt{mctrl} may change over time. To use the post-region of \abs{m}, we must ensure that there is \emph{always} some controller for \abs{mctrl}. To do so, for each object and each method that is specified as being externally controlled, we keep track of the current controller using a type-and-effect system.

A type-and-effect system is a generalization of data type systems, which are defined for some specific side-effect.
It checks the correctness of evaluation for a certain set of objects with respect to this side-effect. In our case, the side-effect of interest is time advance. By keeping track of how much time each statement takes, we can verify whether \abs{mctrl} is called frequently enough. Additionally, we keep track of \emph{ownership}~\cite{DBLP:conf/aplas/ClarkeWOJ08} to ensure that every method that requires to be frequently called is indeed always owned by somebody who does. 
A main feature of our behavioral type system is its parametricity: while we do keep track of effects and ownership, we do not compute how long a certain communication pattern takes. Systems for that kind of property are available~\cite{GiachinoJLP15,DBLP:journals/jlp/LaneveLPR19} or are straightforward to extend; we integrate them through oracles that encapsulate their analysis. That drastically simplifies our system and allows us to focus on the presentation of the novel features of the type-and-effect system.

\paragraph{Contributions}
Our main contribution is a modular deductive verification system for Hybrid Active Objects with \emph{delegated} control that uses a type-and-effect system to govern interactions (1) between multiple Hybrid Active Objects and (2) between Hybrid Active Objects and cloud-models using Timed Active Objects.

\section{Hybrid Active Objects and Post-Regions}
\label{sec:habs}
In this section we present the preliminiaries for the rest of the article.
First, we present the \texttt{HABS} language that implements Hybrid Active Objects.
It is introduced, and fully described, by Kamburjan et al.~\cite{arxiv}.
Here, we only present the language parts that are relevant for post-region based verification and omit, e.g., inheritance, method visibility and variability.

A Hybrid Active Object (HAO) is a strongly encapsulated object, i.e., 
no other object, not even from the same class, may access the fields of an instance. Communication between HAOs is only possible through asynchronous method calls and synchronization: each method call generates a container called \emph{future} for the caller that uniquely identifies the (to be) started process at callee side. The future may be passed around and permits to synchronize (i.e., wait until the called process terminates) with it and read the return value of the associated process.
HAOs implement \emph{cooperative scheduling}:
Inside an object, only one process is active at a time. This process cannot be preempted by the scheduler --- 
it must explicitly release control by either terminating or suspending (via \abs{await}). 
These two properties make (Hybrid) Active Objects easy to analyze: 
Approaches for sequential program analyses can be applied
between two \abs{await} statements (and method start and end).

\begin{figure}[b!t]
\noindent\begin{abscode}
class Tank(Log log){                //class header with field declaration
  physical Real level = 5;          //physical field declaration
  Real drain = -1;                 //field declaration
  physical{ level' = drain; }       //dynamics
  { this!up(); this!down(); }$\label{line:init}$        //constructor
  Unit down(){                     //method header
    await diff level <= 3 & drain <= 0;
    log!triggered(); drain =  1; this!down();
  }
  Unit up(){                       //method header
    await diff level >= 10 & drain >= 0;
    log!triggered(); drain = -1; this!up();
  }
}
\end{abscode}
\caption{A water tank in \HABS with event-based control.}
\label{fig:bball}
\end{figure}

Hybrid Active Objects differ from standard Active Objects by a \abs{physical} block and \abs{physical} fields.
A \abs{physical} field is a field that has some dynamics, while the \abs{physical} block describes the very dynamics of all \abs{physical} fields as ODEs.
These dynamics are used to update the state whenever time advances. 
Inside a method, an imperative language is used, which has special statements to advance time or to wait until some condition on the state holds.
These conditions define an urgent transition: The method is reactivated as soon as possible once the condition holds (and no other process is active). 

\begin{example}\label{ex:eventroom}
Consider the water tank model in Fig.~\ref{fig:bball}.
The tank keeps a water level between \EK{3 liters and 10 liters}.
The pictured class, \abs{Tank} has two discrete fields (\abs{log} and \abs{drain}) and
a \abs{physical} field \abs{level}.
A physical field is described by its initial value and \EK{an} ODE in the \abs{physical} block, which models that the water level is linear with respect to the drain.
Line~\ref{line:init} gives the constructor in form of an initialization block where the two methods 
\abs{up} and \abs{down} are called.
Each method starts with a statement that has as its guard the condition when the process will be scheduled (for \abs{up}, at the moment the level reaches 10 while water rises). 
This is logged by calling the external object  \abs{log} on method \abs{triggered}.
This method call is asynchronous, i.e., the execution of the \abs{up} (or \abs{down}) continues without waiting for it to finish. No future is used in this example. 
Then, the drain is adjusted and the method calls itself recursively to react the next time.
\end{example}

Example~\ref{ex:eventroom} illustrates \emph{event-based} control, as the guard of the \mbox{\abs{await diff}} statement define an event boundary. Alternatively, one may use \emph{time-based} control, as the following example illustrates.

\begin{example}\label{ex:timedroom}
The controller in Fig.~\ref{fig:ctank} checks the water level of a tank once every time unit by using the \abs{await duration} statement to suspend the \abs{ctrl} process for the required amount of time. \EK{We use JML~style~\cite{JML-Ref-Manual} comments for specification.}
\begin{figure}[b!th]
\begin{abscode}
/*@ requires 4 <= inVal <= 9  @*/
/*@ invariant 3 <= level <= 10 && -1 <= drain <= 1 @*/
class TankTick(Real inVal){
  physical Real level = inVal;
  Real drain = -1;
  physical{ level' = drain; }
  { this!ctrl(); }
  Unit ctrl(){
    await duration(1);
    if(level <= 4) drain =  1;
    if(level >= 9) drain = -1;
    this!ctrl();
  }
}
\end{abscode}
\caption{A specified water tank in \HABS with time-based control.}
\label{fig:ctank}
\end{figure}
\end{example}

\subsection{Syntax}

The syntax of \HABS is given by the grammar in Fig.~\ref{fig:syntax}.
Standard expressions \texttt{e} are defined over fields \texttt{f}, variables \texttt{v} and
operators \texttt{!}, \texttt{|}, \texttt{\&}, \texttt{>=}, \texttt{<=}, \texttt{+}, \texttt{-},
\texttt{*}, \texttt{/}.  
Ordinary differential expressions (ODE) are equalities over expressions extended with a derivation operator \abs{e'}.
Types \texttt{T} are all class names, type-generic
futures \abs{Fut<T>}, \abs{Real}, \abs{Unit} and \abs{Bool}. 

\begin{figure*}[tbh]
\resizebox{\textwidth}{!}{
\begin{minipage}{1\textwidth}
\begin{align*}
\mathsf{Prgm} ::=~& \many{\mathsf{CD}}~\{\mathsf{s}\}\quad
\mathsf{CD} ::= \xabs{class C}\left[\xabs{(}\many{\xabs{T f}}\xabs{)}\right]\{\many{\mathsf{FD}}~[\mathsf{Phys}]~[\{\mathsf{s}\}]~\many{\mathsf{Mt}}\} 
&& \text{\small Classes}\\
\mathsf{Mt} ::=~& \xabs{T m(}\many{\xabs{T}~\xabs{v}}\xabs{)}~\{\mathsf{s}\xabs{;}\xabs{return e;}\}
&& \text{\small Methods}\\
\mathsf{FD} ::=~& \xabs{T f[ = e];} \synsep  \xabs{physical Real f = e;}
&& \text{\small Fields}\\
\mathsf{Phys} ::=~& \xabs{physical}~\{\many{\mathsf{ODE}}\} 
&&\text{\small Physical Block}\\
\mathsf{s} ::=~&
\xabs{while (e)}~\{\mathsf{s}\}\synsep
\xabs{if (e)}~\{\mathsf{s}\}~[\xabs{else}~\{\mathsf{s}\}]\synsep
\mathsf{s}\xabs{;}\mathsf{s} && \\
&\synsep\xabs{await}_\xabs{p}~\mathsf{g} \synsep
[[\xabs{T}]~\xabs{e}] = \mathsf{rhs} \synsep \xabs{duration(e)}&&\text{\small Statements}\\
\mathsf{g} ::=~&\xabs{e?} \synsep \xabs{duration(e)} \synsep \xabs{diff e}
&&\text{\small Guards}\\
\mathsf{rhs} ::=~&
\xabs{e} \synsep \xabs{new C(}\many{\xabs{e}}\xabs{)} \synsep 
\xabs{e.get} \synsep
\xabs{e!m(}\many{\xabs{e}}\xabs{)}
&&\text{\small RHS Expressions}
\end{align*}
\end{minipage}
}
\caption{\HABS grammar. Notation $[\cdot]$ denotes optional elements and $\many{~\cdot~}$ lists.}
\label{fig:syntax}
\end{figure*}

A program consists of \EK{a set of classes and a main block}. Each class may have a list of discrete fields that are passed as parameters on object creation and a list of internally declared fields with.
An internally declared field\footnote{\emph{All} fields, independent of where they are declared, are accessible only from their object.} may be declared as \abs{physical}. In this case it must be of \abs{Real} type and  must be initialized. 
Furthermore, a class has a physical block, which defines the dynamics of \abs{physical} fields and must be present if at least one field is \abs{physical}. An optional
initialization block is executed directly after object creation and serves as the constructor. Lastly, a class has a set of methods. 

Methods, initializing and main \EK{blocks} consist of statements. Besides the asynchronous method calls (\texttt{e!m()}) described above, only the following constructs are non-standard:
\begin{itemize}
\item The \abs{duration(e)} statement advances time by \texttt{e} time units. \EK{No other process
may execute in that object during this time lapse.}
\item The \abs{e.get} right-hand side expression reads from a future. A future is a container that is generated by an asynchronous call. 
Afterwards a future may be passed around. With the \abs{get} statement one can read the return values once the called process terminates. Until then, the reading process blocks and no other process can run on the object (that is attempting to read).
\item The $\xabs{await}_\mathtt{p}~\mathtt{g}$ statement suspends the process until the guard \texttt{g} holds. 
A guard is either (1) a future poll \texttt{e?} that \EK{waits until the process for the future in \texttt{e} has finished its computation}, (2) a \abs{duration} guard that advances time, or (3) a differential guard \abs{diff e} that holds once expression \texttt{e} evaluates to true.
\EK{Each such statement has a (program-wide unique) suspension point identifier $\mathtt{p}$, which we use to identify the most recent suspension in a trace.}
\end{itemize}

We assume that all methods are suspension-leading, i.e., each method starts with an \abs{await} statement. This is easily achieved by adding \abs{await diff true} if a method is not suspension-leading without significant changes to the behavior\footnote{The difference is that the process is scheduled and descheduled immediately at its start.}.
Concerning the \abs{physical} block, we only admit blocks describing trivial behavior for all non-\abs{physical} fields.
Finally, we only consider well-typed (w.r.t.\ data types) programs, where differential guards contain only \abs{Real}-typed variables.

\subsection{Semantics}
The runtime semantics is a transition system of the form \mbox{$\mathit{tcn}_1 \rightarrow \mathit{tcn}_2$},
where the configurations $\mathit{tcn}$ have the form $\mathsf{clock}(t)~\mathit{cn}$ for some $t>0$ and
a configuration $\mathit{cn}$, which consists of objects and processes.
For readability's sake, we give the full formal definition in the appendix, as the exact formalization of state adds no further insights here, and only define runtime objects formally.

\begin{definition}[Runtime Objects]
A runtime object has the form 
\[(o, \rho, {\mathsf{ODE}, f}, \mathit{prc}, q)\ \]
\end{definition}
An Object has
an identifier $o$, an object store $\rho$ \EK{that maps the object fields to their values}, the current dynamic
$f$, an active process $\mathit{prc}$ and a set of inactive
processes $q$ as its parameters.  The physical behavior description $\mathsf{ODE}$ is taken from the class declaration. The full runtime syntax of \texttt{HABS} is given it Def.~\ref{def:rsyntax}.


\paragraph{Runs}
The semantics of a programs is expressed as a set of \emph{runs}. A run generated by the operational semantics. For each run, we also generate a set of traces, one per object. 

A trace $\theta$ is a mapping from $\mathbb{R}^+$ to states, meaning that at time $t$ the state of the program is $\theta(t)$.
A trace is extracted from a run by interpolating between two configurations resulting from discrete steps using the last solution.
We say that
$\mathsf{clock}(t_i)~\mathit{cn}_i$ is the final configuration at
$t_i$ in a run, if any other timed configuration
$\mathsf{clock}(t_i)~\mathit{cn}_i'$ is before it.
\begin{definition}[Traces]\label{def:trace}
The initial configuration of a program $\mathtt{Prgm}$ is denoted $\mathit{cn}_0$~\cite{BjorkBJST13}.
A run of $\mathtt{Prgm}$ is a
(possibly infinite) reduction sequence
\[\mathsf{clock}(0)~\mathit{cn}_0 \rightarrow \mathsf{clock}(t_1)~\mathit{cn}_1 \rightarrow \cdots\]
A run is \emph{time-convergent} if it is infinite and $\lim_{i \mapsto \infty} t_i < \infty$.
A run is \emph{locally terminated} if every process occurring within the run terminates normally.

For each object $o$ occurring in the run, its \emph{trace} is defined as $\theta_o$: 

\scalebox{0.9}{\begin{minipage}{\columnwidth}
\begin{align*}
\theta_o(x) = \cased{
\mathit{undefined} &\text{if $o$ is not yet created}\\
\rho &\text{if $\mathsf{clock}(x)~\mathit{cn}$ is the final configuration at $x$}\\ 
&\text{ and $\rho$ is the \EK{store} of $o$ in $\mathit{cn}$.}\\
\mathit{adv}_{\mathit{heap}}(\rho,f,x-y) &\text{if there is no configuration at $\mathsf{clock}(x)$}\\
&\text{and the last configuration was at $\mathsf{clock}(y)$}\\
&\text{with state $\rho$ and dynamic $f$ }
}
\end{align*}
\end{minipage}}

with the following auxiliary function to advance the \EK{store} $\rho$ by $t$ time units according to dynamics $f$.
\[
\mathit{adv}_\mathit{heap}(\rho,f,t)(\xabs{f}) =
\cased{
\rho(\xabs{f}) & \text{ if \abs{f} is not physical} \\
f(t)(\xabs{f})  & \text{ otherwise }}
\]
\end{definition}
The full definition is given in Fig.~\ref{fig:tsem} and illustrated in Ex.~\ref{ex:app} in the appendix.
We normalize all traces and let them start with 0 by shifting all states by the time the object is created.

\begin{example}
    Consider Ex.~\ref{ex:eventroom} and an object where the initial value is 5, i.e., \abs{inVal = 5}.
    It evaluation has the first discrete steps at time 0, 1. The state after the transition is as follows:
    \begin{align*}
        t = 0 &\quad \{\xabs{level} = 5, \xabs{drain} = -1\} \\
        t = 1 &\quad \{\xabs{level} = 4, \xabs{drain} = 1\} \\
        t = 8 &\quad  \{\xabs{level} = 10, \xabs{drain} = -1\} 
    \end{align*}
    The trace $\theta$ thus has the following properties at these times (as per the second case in the above definition:
    \begin{align*}
        \theta(0) &= \{\xabs{level} = 5, \xabs{drain} = -1\} \\
        \theta(1) &= \{\xabs{level} = 4, \xabs{drain} = 1\} \\
        \theta(8) &= \{\xabs{level} = 10, \xabs{drain} = -1\} 
    \end{align*}
    The general solution of the dynamics is 
    \begin{align*}
        \xabs{level}(t) &= \xabs{level}(t_0)+\xabs{drain}(t_0)*t\\
        \xabs{drain}(t) &= \xabs{drain}(t_0)
    \end{align*}
    This is used to define the value of $\theta$ in between. For example for $0 < x < 1$ we have 
    $\xabs{level}(t_0) = \xabs{level}(0) = 5$ and, 
    thus
    \[\theta(x)(\xabs{level}) = \xabs{level}(x-0) = 5+\theta(x-0)(\xabs{drain})*x=5-1*x\]
\end{example}

As we will see later, we must be able to soundly overapproximate the states after a suspension and before the next process is scheduled. To make this precise, we use the notion of \emph{suspension-subtraces}.
\begin{definition}[Suspension-Subtraces]
Let \xabs{C.m} be a method in some program $\mathtt{Prgm}$.
Let $\trace_o$ be a trace, stemming from some run of $\mathtt{Prgm}$ for some object $o$ of class \abs{C}. 
Let $i$ be the index in $\trace_o$ where some process of $\mathtt{m}$ suspends and terminates.
We say that $\trace_o^i$ is the suspension-subtrace of $\trace_o$, if it starts at $i$ and ends at (including) the time where the next non-trivial\footnote{I.e., a process that performs any action instead of descheduling immediately.} process is scheduled. If there is no such time, then $\trace_o^i$ is infinite.
Additionally, $\trace_o^i$ has a variable $\mathtt{t}$ with $\trace^i_o(0)(\mathtt{t}) = 0$ and $\mathtt{t}' = 1$. I.e., a clock that keeps track of the length of $\trace_i$.

The set of all suspension-subtraces of $\mathtt{m}$ in $\mathtt{Prgm}$ is denoted  $\Theta(\mathtt{m},\mathtt{Prgm})$.
\end{definition}
Suspension subtraces are exactly the traces between two discrete steps with length $>0$.
They contain the states where time advances and no process is active (for a given object).
In the above example, $\theta$ has one suspension-subtraces. It defined by the restriction of the domain to $0\leq t \leq 1 $. For $\theta(1)$.



\label{sec:intern}
\subsection{Differential Dynamic Logic}
To verify \HABS, we generate proof obligations that encode that a certain statement or physical process has a certain post-condition. Our logic of choice is differential dynamic logic (\ddl)~\cite{DBLP:journals/lmcs/Platzer12b,Platzer18}, a first-order dynamic logic embeds \EK{\emph{hybrid algebraic programs}} into its modalities.
Hybrid algebraic programs are defined by a simple \EK{imperative} language, extended with a statement for ordinary differential equations. 
Such a statement evolves the state according to some dynamics for a non-deterministically chosen amount of time. 

\begin{definition}[Syntax of $d\mathcal{L}$]
Let $p$ range over predicate symbols (such as $\doteq,\geq$), $f$ over function symbols (such as $\xabs{+}$) and \abs{x} over variables.
Hybrid algebraic programs $\alpha$, formulas $\phi$ and terms $t$ are defined by the following grammar.
\begin{align*}
\phi ::=~&p(\many{t}) \synsep \neg\phi \synsep \phi \wedge \phi \synsep \exists \xabs{x}.~\phi \synsep [\alpha]\phi\qquad
t ::=f(\many{t}) \synsep \xabs{x} \qquad \mathit{dt} := f(\many{dt}) \synsep t \synsep (t)'\\
\alpha ::=~&\xabs{x :=}~t \synsep \xabs{x := *} \synsep \alpha \cup \alpha \synsep \alpha^\ast \synsep ?\phi \synsep \alpha;\alpha \synsep \{\alpha\} \synsep \many{\xabs{x =}~\mathit{dt}} \& \phi 
\end{align*}
\end{definition}
In the following, we use 
the usual derived connectives ($\rightarrow,\vee,\forall$) for brevity.
Modalities $[\cdot]$ contain hybrid algebraic programs and may be nested using the $?$ operator. All ODEs are autonomous.
The semantics of hybrid programs is as follows:
Program $\xabs{x :=}~t$ assigns the value of $t$ to \abs{x}. 
Program $\xabs{x := *}$ assigns a non-deterministically chosen value to \abs{x}. 
Program $\alpha_1 \cup \alpha_2$ is a non-deterministic choice.
Program $\alpha^\ast$ is the Kleene star.
Program $?\phi$ is a test or filter. It either discards a run (if $\phi$ does not hold) or performs no action (if $\phi$ does hold).
Program $\alpha_1;\alpha_2$ is sequence and $\{\alpha\}$ is a block for structuring. 
Finally, the statement $\many{\xabs{x =}~\mathit{dt}} \& \phi$
evolves the state according to the given ODE in the evolution domain $\phi$ for some amount of time. 
The evolution domain describes where a solution is allowed to evolve, not the solution itself.
The semantics of the first-order fragment is completely standard.
The semantics of $[\alpha]\phi$ is that $\phi$ has to hold in \emph{every} post-state of $\alpha$ if $\alpha$ terminates.
We stress that if $\alpha$ is an ODE, then this means that $\phi$ holds throughout the \emph{whole} solution.

\begin{example}
The following formula expresses that the position of a bouncing ball
with initial position \abs{x} below 10 meters and initial null velocity \abs{v} is below 10 before reaching the ground (given that the gravity is 9.81$m/s$).
\[0 \leq \xabs{x} \leq 10 \wedge \xabs{v} \doteq 0 \rightarrow [\xabs{x}' = \xabs{v}, \xabs{v}' = -9.81 \& \mathtt{x} \geq 0]\xabs{x} \leq 10\]
Events can be expressed as usual by an event boundary created between a test and an evolution domain. 
The following program models that the ball repeatedly bounces back exactly on the ground.
\[\big(\{\xabs{x}' = \xabs{v}, \xabs{v}' = -9.81 \& \mathtt{x} \geq 0\}; ? \mathtt{x} \leq 0; \xabs{v} := -\xabs{v}*0.9\big)^\ast\]
\end{example}
We identify \HABS variables and fields with $d\mathcal{L}$ variables and denote with
$\mathsf{trans}(\xabs{e})$ the straightforward translation of \HABS expressions into $d\mathcal{L}$ terms. 
Standard control flow constructs ($\mathsf{while}$, $\mathsf{if}$) are encoded using the operators above~\cite{DBLP:conf/lics/Platzer12b}.

\EK{Weak negation $\tilde\neg$ is needed to define event boundaries. It is defined analogously to normal negation, except for weak inequalities:
\[\tilde\neg(t_1 \geq t_2) \equiv t_1 \leq t_2\]
}

\subsection{Post-Region Invariants}\label{ssec:post}
To verify an object invariant, one generates a proof obligation in dynamic logic for each method, and one for the constructor. If all proof obligations can be closed, i.e., the dynamic logic formulas are all valid, then the object invariant holds at every point a process starts, ends, suspends or regains control. This approach is modular, as changes is one method do not require to reprove other methods.

There are such systems for numerous discrete object-oriented languages, e.g., Java (in the KeY-system using Java Dynamic Logic~\cite{key}) and \texttt{ABS} (in KeY-ABS~\cite{DinBH15} using ABS Dynamic Logic~\cite{DinO15} and in Crowbar~\cite{DBLP:journals/corr/abs-2102-10127} using Behavioral Program Logic~\cite{bpl}).
In the most basic case the proof obligations for an invariant $I$ take the following form for discrete languages:
\[ I \rightarrow [\mathtt{s}]I \]
where $\mathtt{s}$ is the method body of the method in question; and for the constructor
\[\mathsf{true} \rightarrow [\mathtt{s}]I\quad.\]
The main idea is that the constructor always establishes the object invariant and each method preserves it. Each method may assume the invariant, because the last process established it and in discrete system, \emph{state does not change} when no process is active. This is not the case for hybrid systems, the above proof obligation scheme is \emph{not} sound.

To accommodate hybrid systems the proof obligation scheme must incorporate the dynamics in the post-condition, as a so called \emph{post-region invariant}~\cite{Kamburjan21a}. The case for methods is the following:
\[ I \rightarrow \big[\mathtt{s}\big]\big(I \wedge [\mathsf{dyn} \& \phi]I\big) \]
where $\mathsf{dyn}$ are the dynamics and $\phi$ is the \emph{post-region}.
The post-region is the region where the dynamics must be safe. 
We say that $I$ is the post-region invariant for $\phi$ and stress that it is necessary to establish $I$ as a pure post-condition as well -- it may be the case that $\phi \equiv \mathsf{false}$, i.e., that the next process starts \emph{without} time advance\footnote{\EK{The post-region is a part of the state space of the object, with time as a dimension.}}. To ensure that this next process can also assume $I$, it is necessary to add $I$ without dynamics to the post-condition.

If e \emph{basic} post-region is just $\mathsf{true}$, i.e., the dynamics must stay safe forever. 
In general, basic post-regions are not sufficient -- consider the two models in Ex.~\ref{ex:timedroom} and Ex.~\ref{ex:eventroom}: these systems are not safe for an unlimited time, instead there are internal control loops that define when a discrete computation will start. I.e., it suffices to restrict $\phi$ to the region where it is \emph{not} guaranteed that another method will take over. One can easily extend the precondition, if the method starts with a guard, by adding the guard to the left-hand side of the implication.

Two further possible ways to soundly compute more precise \emph{internal} post-regions were proposed: structural control and (method-)local control~\cite{Kamburjan21a}. They have in common that they are local -- the post-region is computed based on information derived from a single class. They cannot, however, verify the above examples.

Next, we define the formalization of general soundness for post-regions~\cite{Kamburjan21b}, which parameterizes the proof obligation scheme with post-region generators. 
%
In the following, we denote the specified invariant for a class \abs{C} with $I_\xabs{C}$. For initialization, a constraint on the initial values of the externally initialized fields may be specified. This \emph{creation condition} is denoted $\mathsf{pre}_\xabs{C}$ and used a precondition for the constructor.

\subsection{General Proof Obligation Scheme}
A proof obligation scheme defines a set of \ddl-formulas, such that validity of all these formulas implies safety of the program. The scheme we give here is parametric \EK{in} the notion of post-region, as well as in the specification. Method contracts, \EK{in the sense of pre-/postcondition pairs,} are not of interest here, we only use a precondition $\mathsf{pre}_\xabs{C.m}$, which is a first-order formula over the method parameters, and a postcondition $\mathsf{post}_\xabs{C.m}$, which is a first-order formula over the fields of the class.
\EK{Similarly, $\mathsf{init}_\xabs{C}$ is the precondition of the initial block/constructor and $I_\xabs{C}$ is the class invariant.}

\EK{Before we define the proof obligation scheme, we must establish some auxiliary structures.}
\begin{itemize}
    \item 
We assume two variables \abs{t} and \abs{cll}. Variable \abs{t} keeps track of time and variable \abs{cll} keeps track of contract violations. This is necessary, because the post-condition is evaluated at the end of the methods and intermediate failures must be remembered until then.

    \item The $\mathsf{fail}$ statement sets \abs{cll} to 1, i.e., records a contract violation.
    \[\mathsf{fail} = \xabs{cll :=}~1\]
    \item The $\mathsf{havoc}$ statement sets all fields, including all physical fields, \emph{but not local variables} to new values.
    This is used to approximate suspension, where another process can run, but only change fields.
    \[\mathsf{havoc} = \xabs{f}_1\xabs{:= *;}\dots\xabs{f}_n\xabs{:= *;}\quad\text{\it for all fields $\xabs{f}_i$ }\]
    \item The $\mathsf{havoc}^\mathsf{ph}$ statement sets all \emph{physical} fields to new values.
    \[\mathsf{havoc}^\mathsf{ph} = \xabs{f}_1\xabs{:= *;}\dots\xabs{f}_n\xabs{:= *;}\quad\text{\it for all physical fields $\xabs{f}_i$ }\]

\item \EK{The post-region formula} $\mathsf{pr}(\phi, I, \mathsf{ode})$ expresses that a certain invariant $I$ holds for dynamics $\mathsf{ode}$ under post-region $\phi$:
\[
\mathsf{pr}(\phi, I, \mathsf{ode}) = I \wedge \big[\xabs{t} \xabs{:=} 0;\{\mathsf{ode},\xabs{t}' = 1\& \phi\}\big] I\\
\]
\end{itemize}
Whenever $I$ and $\mathsf{ode}$ are understood, we just write $\mathsf{pr}(\phi)$.
Next, we define the proof obligation scheme itself.
\begin{definition}[Proof Obligation Scheme]\label{def:scheme}
Let \abs{Prgm} be a program. For each class \abs{C} in the program, let $\mathsf{ode}_\xabs{C}$ be its dynamics and $\statement_\xabs{C.init}$ the code of the constructor.
For every method \abs{C.m} let $\statement_\xabs{C.m}$ be the method body.
Let $\psi$ be a family of post-regions, i.e, formulas over the physical fields of a class and \abs{t}, indexed with (1) method names including the constructor, and (2) suspension point identifiers. Let $\statement_\xabs{main}$ be the statement of the main block.
The proof obligation scheme $\iota^\psi$ \EK{for family $\psi$ is a function from methods and initial block to formulas, defined as follows. We use the subscript notation for $\iota$.}
For every class \abs{C}, there is one formula 
\[\iota_\xabs{C.init}^\psi \equiv \mathsf{init}_\xabs{C} \wedge \xabs{cll} \doteq 0\rightarrow \big[\mathsf{trans}(\statement_\xabs{init})\big]\big(\xabs{cll} \doteq 0 \wedge \mathsf{pr}(\psi_\xabs{C.init}, I_\xabs{C}, \mathsf{ode}_\xabs{C})\big)\]
for each method \methodname in \xabs{C} one formula

\noindent\resizebox{\textwidth}{!}{
\begin{minipage}{1.05\textwidth}
\[\iota_\xabs{C.m}^\psi \equiv I_\classname \wedge \mathsf{pre}_{\xabs{C.m}} \wedge \xabs{cll} \doteq 0 \rightarrow \big[\mathsf{trans}(\statement_\xabs{C.m})\big]\big(\xabs{cll} \doteq 0 \wedge
\mathsf{post}_{\xabs{C.m}} \wedge  \mathsf{pr}(\psi_\xabs{C.m}, I_\xabs{C}, \mathsf{ode}_\xabs{C})\big)\]
\end{minipage}
}
and for the main block the formula
\[ \iota_\xabs{main}^\psi \equiv\mathsf{cll} \doteq 0 \rightarrow  \big[\mathsf{trans}(\statement_\xabs{main})\big]\xabs{cll} \doteq 0 \quad.\]
The translation $\mathsf{trans}$ of \HABS statements into \ddl statements is given in Fig.~\ref{fig:trans}.
\begin{figure}[t!bh]
\scalebox{1}{\begin{minipage}{\columnwidth}
\begin{align*}
\mathsf{trans}(\statement_1;\statement_2) &= \mathsf{trans}(\statement_1)\xabs{;}\mathsf{trans}(\statement_2) \\
\mathsf{trans}(\xabs{if(e)}\{\statement_1\}\xabs{else}\{\statement_2\}) &= 
\mathsf{if}(\mathsf{trans}(\xabs{e}))\{\mathsf{trans}(\statement_1)\} \mathsf{else} \{\mathsf{trans}(\statement_2)\}\\
\mathsf{trans}(\xabs{while(e)}\{\statement\}) &= \mathsf{while}(\mathsf{trans}(\xabs{e}))\{\mathsf{trans}(\statement)\} \\
\mathsf{trans}(\xabs{v = e}) &= \xabs{v:=}~\mathsf{trans}(\xabs{e})\\
\mathsf{trans}(\xabs{return e}) &= \xabs{result :=}~\mathsf{trans}(e)
\end{align*}
\begin{align*}
\mathsf{trans}(\xabs{v = e.get}) &= 
\big\{\{?\mathsf{pr}(\mathsf{true})\}\xabs{;}\mathsf{havoc}^{\mathsf{ph}}\xabs{;}?I_\xabs{C}\xabs{;}\xabs{v :=*}\big\} \cup \big\{\{?\neg\mathsf{pr}(\mathsf{true}); \mathsf{fail}\}\xabs{;}\mathsf{havoc}^{\mathsf{ph}}\xabs{;}\xabs{v :=*}\big\}\\
\mathsf{trans}(\xabs{await}_\mathtt{p}~\xabs{g}) &= 
\phantom{\cup}\{?\mathsf{pr}(\psi_\xabs{p} \wedge \tilde{\neg}\mathsf{trans}(\xabs{g}))\}\xabs{;}\mathsf{havoc}\xabs{;}?\mathsf{trans}(\xabs{g}) \wedge I_\xabs{C}\} \\
&\phantom{=}\cup 
\{?\neg\mathsf{pr}(\psi_\xabs{p} \wedge\tilde{\neg}\mathsf{trans}(\xabs{g}))\xabs{;}\mathsf{fail}\}\xabs{;}\mathsf{havoc}\xabs{;}?\mathsf{trans}(\xabs{g})\} 
\end{align*}
\begin{align*}
\mathsf{trans}&(\xabs{v = e!m(e}_1,\dots\xabs{)}) =
\big\{\{?\mathsf{pre}_{\methodname}(\xabs{e}_1,\dots) \} \cup \{?\neg\mathsf{pre}_{\methodname}(\xabs{e}_1,\dots)
\xabs{;}\mathsf{fail} \}\big\}\xabs{;}
\xabs{v := *} \\
\mathsf{trans}&(\xabs{v = new C(e}_1,\dots\xabs{)}) =
\big\{\{?\mathsf{init}_{\classname}(\xabs{e}_1,\dots) \} \cup \{?\neg\mathsf{init}_{\classname}(\xabs{e}_1,\dots)
\xabs{;}\mathsf{fail} \}\big\}\xabs{;}
\xabs{v := *} \\
\mathsf{trans}&(\xabs{duration(e)}) = 
~\xabs{t := 0;}\big\{\{?\mathsf{pr}(\xabs{t} \leq \xabs{e})\} \cup \{?\neg\mathsf{pr}(\xabs{t} \leq\xabs{e}); \mathsf{fail}\}\big\}\xabs{;}\\
&\hspace{22.5mm}\xabs{t := 0;}\{\mathsf{ode},\xabs{t}' = 1 \& \xabs{t} \leq \mathsf{trans}(\xabs{e}) \}\xabs{;} ?\xabs{t} \geq \mathsf{trans}(\xabs{e}) 
\end{align*}
\begin{align*}
\mathsf{trans}(\xabs{f}) &= f\text{ , where $f$ is the \ddl variable representing field \xabs{f}}\\
\mathsf{trans}(\xabs{v}) &= v\text{ , where $v$ is the \ddl variable representing variable \xabs{v}}\\
\mathsf{trans}(\xabs{e}_1~\mathit{op}~\xabs{e}_2) &= \mathsf{trans}(\xabs{e}_1)~\mathit{op}~\mathsf{trans}(\xabs{e}_2)\\
\mathsf{trans}(\xabs{diff e}) &= \mathsf{trans}(\xabs{e})\\
\mathsf{trans}(\xabs{g}) &= \mathsf{true}\text{ if \abs{g} does not have the form \abs{diff e}}\\
\end{align*}
\end{minipage}}
        
\caption{Translation of \HABS-statements into \dL programs.}
\label{fig:trans}
\end{figure}
\end{definition}
If $\psi_{\xabs{C.m}}$ is the same for all methods and the class is understood, then we drop the index.
We first examine the proof obligation for normal methods. 
The precondition $I_\classname \wedge \mathsf{pre}_{\xabs{C.m}} \wedge \xabs{cll} \doteq 0$ expresses that the object invariant and the method precondition hold. The last term initializes the special variable \abs{cll}.
The first term of the post-condition (of the modality) expresses that no intermediate check failed the proof and \abs{cll} is still $0$.
The second term \EK{checks} the method post-condition and the last term ensures safety in the post-region.
As previously discussed, it takes the form $I_\xabs{C} \wedge \left[\xabs{t} \xabs{:=} 0;\{\mathsf{ode},\xabs{t}' = 1\& \phi\}\right] I_\xabs{C}$. It expresses that the invariant holds when the method terminated and that it is an invariant for the dynamics in a defined post-region. 

The proof obligation for constructors is analogous but (1) does not assume the invariant, as it has not been established yet, and (2) assumes the creation condition as its precondition. A constructor has no post-condition. The proof obligation for the main block is \emph{only} checking that all calls and object creations adhere to the respective precondition, as it runs outside of any object and, thus, has no additional specification.

The translation of statements into \ddl-programs works as follows.
We consider all fields as variables for the translation.
The translation of sequence, branching, loops and assignment of side-effect free expressions to location is straightforward. We can ignore the expression of \abs{return} statements as invariants cannot specify return values. The other statements are translated as follows:
\begin{itemize}
    \item Synchronization with \abs{get} first checks $\mathsf{pr}(\mathsf{true})$, if the formula does not hold, then verification fails. This models that during synchronization, time may pass and the invariant must, thus, hold. It is \emph{not} sound to assume the post-region $\psi$ here: synchronization blocks, so no other process can run. Furthermore, it may stay blocked for an unbound amount of time, so the invariant must hold for an unbound amount of time as well.
    Afterwards, \abs{v} is set to a new, unknown value, as the return value in a future is not specified. Additionally, $\mathsf{havoc}^\mathsf{ph}$ is used to model that the physical fields may have changed during the synchronization. In case the check succeeds, the invariant is known to hold for the new values of the fields.
    \item Suspension with \abs{await} is similar, but uses both the post-region $\psi$ \emph{and} the guard.
    We use $\mathsf{havoc}?\xabs{;}\mathsf{trans}(\xabs{g})$ to set all fields (but not variables) to new values -- contrary to the case of \abs{get}, the non-physical fields may have been changed by another process.
    For the new values only the guard is known to hold. The invariant is also known to hold, but only if the check succeeds. We stress again that time advancement is modeled in the contained modality.
    \item Method calls check the precondition of the called method. Again, \abs{v} is set to a new, unknown value to model a fresh future.
    \item Object creation is analogous to method calls.
    \item Finally, blocking time advance is similar to synchronization using \abs{get}, with two differences: 
    First, while it is still not sound to use $\psi$, we may limit the time spent executing this statement.
    Second, instead of causing $\mathsf{havoc}$, we can precisely simulate the state change by advancing the dynamics for the amount of time given in the \abs{duration} statement.
\end{itemize}

A scheme generates one formula per entity. Its ultimate aim, however is to establish a safety property for the overall system. Indeed, if we use the post-region \texttt{false} for all methods, we may be able to show validity of all formulas -- yet it does mean that the system is safe. We, thus, need a formalization of
the conditions when the validity of all proof obligations generated by a scheme imply safety.

Formally, a scheme is sound if validity of all generated formulas is sufficient to prove safety of all class with respect to their invariants, and safety of all methods with respect to their contracts. We consider partial correctness~\cite{DBLP:journals/cacm/Hoare69}, i.e., we do not consider deadlock and non-terminating programs.
\begin{definition}[Sound Proof Obligation Scheme]\label{def:soundscheme}
If the validity of all proof obligations from $\iota^\psi$ implies that for all 
locally terminating, time-divergent runs, 
$\mathsf{inv}_{\xabs{C}}$ holds in every state of every trace of every object $o$ realizing any class \xabs{C}, whenever (1) $o$ is inactive or (2) time advances, and that  
the pre-condition of a method holds in every prestate and the post-condition in every poststate,
then we say that $\iota^\psi$ is sound.
\end{definition}
Condition (2) expresses that the object stays safe whenever time advances, even if a method is already active. This is critical, as otherwise an object would be in an unsafe state, but would still be considered safe if it, for example, performs a non-suspending \abs{duration}.

We can break down soundness of the scheme into two parts: it must describe the discrete transitions correctly, and it must describe the suspension-subtraces correctly. For the former, we observe that this property can be shown by reasoning about the translation function $\mathsf{trans}$ -- we, thus, only need a formalization for the later. 

We remind that suspension-subtraces contain the states where time advances and no process is active, and that (non-trivial) post-regions are not used when time advances and a process is active, for example during the execution of an \abs{get} statement.
\begin{definition}[Sound Post-Regions]
Let $\mathbf{Prgm}$ be a set of programs, all containing a class \abs{C} with a method \abs{m}. 
Let $\psi$ be a first-order formula over the fields of \abs{C} and the variables of \abs{m}.
We say that $\psi$ is $\mathbf{Prgm}$-sound for \abs{C.m}, if
every state of every suspension-subtrace of every program in $\mathbf{Prgm}$ is a model for $\psi$:
\[\forall \mathtt{Prgm} \in \mathbf{Prgm}.~\forall \trace \in \Theta(\mathtt{m},\mathtt{Prgm}).~\forall i \leq |\trace|.~\trace[i] \models \psi\]
\end{definition}

This is indeed sufficient -- to show soundness of the proof obligation scheme, it suffices that the used post-region generator is sound: The following theorem~\cite{Kamburjan21b} states that soundness of post-regions implies soundness of proof obligation schemes.
\begin{theorem}[\cite{Kamburjan21b}]\label{thm:main}
If $\psi_\xabs{C.m}$ is $\mathbf{Prgm}$-sound for all \abs{C.m}, then the proof obligation scheme of Def.~\ref{def:scheme} is sound for $\mathbf{Prgm}$ in the sense of Def.~\ref{def:soundscheme}.
\end{theorem}

Basic post-regions are obviously sound. A slightly more complex notion is the one of locally controlled post-regions~\cite{Kamburjan21a}. For a simple example, consider a method \abs{m} without branching or suspension that calls another method \abs{called}. Method \abs{called} has the leading guard \abs{x >= 0}.
Then the post-region for \abs{m} is \abs{x <= 0} -- it describes all suspension-subtraces until another process runs, namely the one it called itself. Thus, a post-region generator that assigns \abs{x <= 0}  to \abs{m} is sound.

\EK{
Concrete examples of post-regions that are able to verify Ex.~\ref{ex:eventroom} are described in prior work~\cite{Kamburjan21a}.
In the next section, we introduce a similar system that is not verifiable with those post-regions.}
\section{Externally Controlled Timed Post-Regions}
\label{sec:externaltime}

In the following we consider timed control, where the controlling discrete process is \emph{outside} the object of the controlled physical process where the post-region is to be used. 
To retain modularity of the proof system, we aim to keep the proof obligations the same as before, but instead of \emph{deriving} that a method implements a timed controller, 
we require the user to \emph{specify} it. The overall system then has to ensure that this method is indeed \emph{globally} called with the required frequency.
This property in turn, is handled by a type system -- it is a structural property of the whole program, and as such inherently non-local. 
By using a lightweight type analysis, we keep the required user interactions during the analysis low. 

As we target a more volatile situation of IoT systems, which often come with cloud components, we allow the controlling discrete process
to change. For example, we allow one controller to shut down and another to take over.
We also allow multiple controlling discrete processes to control different aspects of a physical controller,
 e.g. an internal controller for event-based properties and an external controller for timed-based ones. 

Before we come to the formal details, we illustrate the targeted kind of system with a smaller example.
We use again JML~style comments for specification.

\begin{example}\label{ex:externaltimedhao}
Consider the upper code in Fig.~\ref{fig:room}, a variation of the timed water tank of Ex.~\ref{ex:timedroom}.
The \abs{/*@ requires ... @*/} clause specifies the creation condition and \abs{/*@ invariant ... @*/} specifies the safety invariant.
The \abs{Tank} class has the same physical behavior as before, but the \abs{ctrl} method is replaced by \abs{localCtr}
 which does \emph{not} repeatedly perform the check on its own. 
Instead, it is specified with \abs{/*@ timed_requires 1 @*/} that the method must be called at least once per time unit.

In this example, there is only one object of class \abs{Tank}, line~\ref{ex:type:ln:tank},
and the responsibility of calling \abs{localCtrl} on this object is then shared between the methods \abs{Mobile.run} and \abs{Controller.timer}.
Since the method \abs{Mobile.run} creates the tank, it become by default the initial controller to all its controlled methods.
On the other hand, the \abs{Controller.timer} method is annotated with the \abs{/*@ time_control: t.localCtrl = [1, 0] @*/} clause, which means that
 this method takes control of the method \abs{localCtrl} of its parameter \abs{t},
  waits for 1 unit of time (with the \abs{await duration(1)} statement line~\ref{ex:type:ln:await}),
  calls \abs{localCtrl} line~\ref{ex:type:ln:call}, recursively calls itself line~\ref{ex:type:ln:rec} and stops,
 leaving 0 unit of time until the next call to \abs{t.localCtrl}.

Hence, upon calling \abs{Mobile.run}, the \abs{Tank} object \abs{t} is created,
 a \abs{Controller} object \abs{c} is created, and the control of \abs{t.localCtrl} is directly transferred to \abs{c.timer}.
After 40 time units, the \abs{Mobile} instance synchronizes with \abs{c.timer} and a new \abs{Controller} takes over the control of \abs{t.localCtrl} forever.

There is subtle, timing related bug in this code.
At $t=40$, the final call to \abs{timer} does not result in a call to \abs{localCtrl}, but the \abs{await} statement is still executed
 making time advancing to $t=41$ before the method's termination.
As the newly created controller also waits for one second at the beginning of its execution, \abs{localCtrl} next call is a $t=42$:
 the required call at $t=41$ is skipped.
The lower code in Fig.~\ref{fig:room} gives a solution: by only advancing time when a call is made afterwards, the gap at $t=41$ can be avoided.

Note that this bug can be identified in the specification of the the faulty version of \abs{timer} method.
Indeed this specification states that the method:
 $i)$ waits 1 unit of time at the beginning of its execution before calling \abs{t.localCtrl};
 and $ii)$ concludes its execution with \abs{t.localCtrl} having to be called right away.
Hence it is unsound to delegate the control of the tank to sequences of calls to \abs{timer}.
The fixed version of the \abs{timer} method concludes it execution with \abs{t.localCtrl} having to be called after 1 time unit,
 and so sequences calls to \abs{timer} do correctly control the tank.

Such subtle bugs illustrate both the need for tool support in the analysis of distributed hybrid and timed systems, as well as the value of specification.
\begin{figure}
\begin{abscode}
/*@ requires 4 <= inVal <= 9  @*/
/*@ invariant 3 <= level <= 10 && -1 <= drain <= 1 @*/
class Tank(Real inVal){
  physical Real level = inVal;
  Real drain = -1;
  physical{ level' = drain; }
  /*@ timed_requires 1 @*/
  Unit localCtrl(){
    if(level <= 4) drain =  1;
    if(level >= 9) drain = -1;
  }}

class Controller(){
  /*@ requires t!= null && time_control: t.localCtrl = [1, 0] @*/
  Unit timer(Tank t, Int time){
    await duration(1); $\label{ex:type:ln:await}$
    if(time != 0) { 
      t!localCtrl(); $\label{ex:type:ln:call}$
      Fut<Unit> f = this.timer(t, time - 1); $\label{ex:type:ln:rec}$
      await f?;
  }}}

class Mobile {
  Unit run() {
    Tank t = new Tank(4); $\label{ex:type:ln:tank}$
    Controller c = new localCtrl(); Fut<Unit> f = c.timer(t, 40); $\label{ex:type:ln:call40}$
    await duration(40);     await f?; $\label{ex:type:ln:await40}\label{ex:type:ln:awaitf}$
    c = new Controller(); f = c.timer(t, -1); $\label{ex:type:ln:newc}$
}}
\end{abscode}
\begin{abscode}
class Controller(){
  /*@ requires t!= null && time_control: t.localCtrl = [1, 1] @*/
  Unit timer(Tank t, Int time){
    if(time != 0) { 
      await duration(1); $\label{ex:type:ln:await2}$
      t!localCtrl(); $\label{ex:type:ln:call2}$
      Fut<Unit> f = this.timer(t, time - 1); $\label{ex:type:ln:rec2}$
      await f?;
  }}}
\end{abscode}
\caption{An externally controlled tank with mobile control. The upper version of \abs{Controller.timer} contains a subtle bug regarding timing, which is fixed in the lower version.}
\label{fig:room}
\end{figure}
\end{example}

Let us call \coreid\ a pair of a location (a variable or a field) and a method, such as \abs{t.localCtrl} or \abs{c.timer} in our example.
There are several structural requirements that need to be checked to ensure that a control pattern such as the one presented in our example works:
 (1) for each controlled \coreid, there is always a controlling \coreid;
 (2) the controlling \coreid\ is indeed observing the specified time behavior and (3) if the controlling \coreid\ changes, there are no gaps in control.
If all these properties can be ensured, then the specification of the marked timed controller method can be used in the post-region.

The proof obligations do not change: the post-region for the methods in \abs{Tank} are defined by the frequency of the timed controller.
The proof obligations needed for the specification of \abs{Controller} are not hybrid (as the class contains no \abs{physical} block) and can be handled by discrete approaches to Active Object verification.
\begin{example}\label{ex:tank}
For \abs{Tank}, the formulas in Fig.~\ref{fig:obli} are generated.
Let $I$ be the invariant specified in Fig.~\ref{fig:room}, $\psi$ the mapping from each method to $\mathtt{t} \leq 1$ and $\mathit{dyn} \equiv \xabs{level' = drain}$. 
\begin{figure}
\resizebox{\textwidth}{!}{%
\begin{minipage}{\linewidth}
\begin{align*}
    &\iota^\psi_{\xabs{Tank.init}}\\ \equiv~&4 \leq \mathtt{inVal} \leq 9 \rightarrow \big[ \mathtt{level := inVal; drain := -1;} \big]\big(I \wedge [\mathtt{t := 0}\{\mathit{dyn}, \mathtt{t'} = 1 \& \mathtt{t} \leq 1\}]I\big)\\
    &\iota^\psi_{\xabs{Tank.localCtrl}} \\\equiv~ &I \rightarrow  \big[ \mathtt{if(level \leq 4) drain := 1;if(level \geq 9) drain := -1;}  \big]\big(I \wedge [\mathtt{t := 0}\{\mathit{dyn}, \mathtt{t'} = 1 \& \mathtt{t} \leq 1\}]I\big)
\end{align*}
\end{minipage}
}
\caption{Proof obligations for \abs{Tank} in Ex.~\ref{ex:tank}.}
\label{fig:obli}
\end{figure}
\end{example}

\subsection{Type System}
\label{subsec:ts}

As discussed previously, this type system has one unique goal:
 check that methods are called correctly with respect to their {\commentfont timed\_requires} annotation.
It is thus entirely independent from the physical aspects of {\tt HABS} and focuses only on the time aspect of a {\tt HABS} program.
In particular, this type system must, to reach its goal,
 perform a {\em time analysis} of the input program, i.e.,
 compute how much time each statement can take (in particular the \abs{await} and \abs{duration} statements).
Then it must use the information provided by this time analysis
 to keep track of the \coreid\ control relationship and ensure that all \coreid\ are correctly called.

\subsubsection{Time Analysis.}
Designing a time analysis is a difficult task, since such an analysis is undecidable in general (it includes the halting problem), yet the concrete design choices are not central to this work.
Many such design choices must be made to decide which behaviors of a program is abstracted away by the analysis,
 and possibly many complex structures and algorithms must be defined to precisely analyze the rest of the program.
Interestingly, resource analysis, and time analysis in particular, have already been defined for {\tt ABS}~\cite{DBLP:conf/pepm/AlbertAGGP12,AlbertBHJSTW14,DBLP:journals/jlp/LaneveLPR19} and it is reasonable to imagine that other time analysis, with different capabilities, will be defined in the future.
In order to take advantage of the existing (and possibly future) time analysis, we design our type system to be able to use any of them:
 our type system is thus parametric, and given any correct time analysis, ensures that methods are called correctly.
The following definition informally describes the different features of a time analysis that are needed by our type system:
\begin{definition}
A {\em Time Analysis} for a given program {\sf Prgm} is a triplet $(\CF,\tacontext,\TA)$ where:
\begin{itemize}
\item $\CF$ are expressions used to describe how time passes.
 Since explicit time advance is expressed with rationals in {\tt HABS}, $\CF$ must include $\mathbb{Q}$,
 and since some computation can take infinite time, expressions in $\CF$ must be comparable to $\infty$.
\item $\tacontext$ is a function that gives information about the {\em execution context} of methods and statements.
 Indeed, since the behavior of a method can change depending on its parameters and the state of the callee,
 it might be relevant for the time analysis to be sensitive to such execution context and give how much time lapses in a statement depending on an execution context.
\item Finally, $\TA$ is the function giving how much time a method or a statement takes depending on the current execution context.
\end{itemize}
\end{definition}

\subsubsection{The \coreid\ Control Relationship.}
Within a method, a \coreid\ can be controlled in two ways:
\begin{enumerate}
\item either a \coreid\ is locally controlled (i.e., the current method is the one responsible to call the \coreid),
 in which case we store how much time is left until a call to the \coreid\ is required;
\item either the control of a \coreid\ has been delegated
 (i.e., the current method passed the control of \coreid\ to a different process and might get the control back later),
 in which case we keep track of which future controls the \coreid, and when that future terminates.
\end{enumerate}
Consequently, a typing statement for {\tt HABS} statements has the form $\Gamma_l,\Gamma_d\typep_c \xabs{s}\typed \Gamma_l',\Gamma_d'$ where:
 $\Gamma_l$ registers the \coreid\ locally controlled;
 $\Gamma_d$ registers the \coreid\ whose control has been delegated;
 $c$ is the execution context given by and forwarded to the time analysis (with the $\tacontext$ and $\TA$ functions);
 $\xabs{s}$ is the typed statement;
 and $\Gamma_l'$ (resp. $\Gamma_d'$) is the locally controlled \coreid\ (resp. delegated \coreid) obtained after executing $\xabs{s}$.
The context $\Gamma_l$ maps \coreid s to the maximum amount of time that can lapse before the method must be called.
The context $\Gamma_d$ maps \coreid s to tuples $(\mathit{fid}, t_{\min},t_{\max},t)$ where:
 $\mathit{fid}$ is the future to which the \coreid\ has been delegated;
 $t_{\min}$ (resp. $t_{\max}$) is the minimum (resp. maximum) amount of time before the future is resolved;
 and $t$ is the maximum amount of time between $\mathit{fid}$ is resolved and the next time \coreid\ must be called.

\begin{remark}\label{rek:no-loop+ssa}
To keep the presentation of our type system simple, we suppose two restrictions on the syntax of the input program \textsf{Prgm}:
\begin{itemize}
\item \textsf{Prgm} does not contain any loop, and
\item every assignment in \textsf{Prgm} declares a new variable.
\end{itemize}
These restrictions do not \EK{limit} the expressivity of \HABS, since loops can be translated into recursive method calls,
 and variables can always be renamed in fresh variables, following a {\em Static Single Assignment} pattern.
\end{remark}

\begin{example}\label{ex:type:new:1}
To get a first impression of how this control relationship works,
 let us look at the method \abs{Mobile.run} in Fig.~\ref{fig:room},
 and look at how its execution using the faulty version of the \abs{Controller.timer} method shape the $\Gamma_l$ and $\Gamma_d$ contexts.
Since the \abs{Mobile} class does not have any field and \abs{Mobile.run} does not have any parameter,
 we can consider that this method does not control anything when it starts:
 $\Gamma_l$ and $\Gamma_d$ are both empty.
Then, after the creation of the tank \abs{t}, which contains a method annotated with {\commentfont time\_requires},
 that method must be locally controlled:
 $\Gamma_l$ now maps the \coreid\ \abs{t.localCtrl} to 1 (i.e., the content of the {\commentfont time\_requires} annotation),
 and $\Gamma_d$ is still empty.
After the creation of the controller \abs{c}, does not not contain any method annotated with {\commentfont time\_requires},
 $\Gamma_l$ and $\Gamma_d$ are left unchanged.
The call to \abs{c.timer(t,40)} does change the contexts however:
 the \abs{timer} method states in its annotation that it takes over the control of \abs{t}.
Here, our type system must first check that the control transfer is sound
(i.e., that the \abs{timer} method will not call \coreid\ too late),
and then registers the transfer in the contexts:
 $\Gamma_l$ becomes empty and $\Gamma_d$ now states that the future \abs{f} controls \coreid.
We use the time analysis to check that \abs{f} controls \coreid\ for 40 units of time,
and so the various \abs{await} statements are correct w.r.t. the control of \coreid.
After the \abs{await f} statement, the control of \coreid\ is once again local:
 $\Gamma_l$ now maps \coreid\ to 0 (i.e., the content of the {\commentfont time\_control} annotation of the \abs{timer} method),
 and $\Gamma_d$ is now empty.
Finally, the creation of a new controller \abs{c} and the call \abs{c.timer(t,-1)} is different from before:
 here the check that the control transfer is sound fails
 since $\Gamma_l$ states that \coreid\ must be called right away and the annotation on \abs{c.timer}
 states that \coreid\ will not be called sooner than after 1 unit of time.
\end{example}

\medskip

The rest of this Section first introduces the definition giving what we consider to be a well-typed program,
 and presents the different rules defining our type system.

\begin{definition}\label{def:type2}
A program \textsf{Prgm} is {\em well-typed} iff all its methods and its \abs{main} can be validated with the rules presented in the following paragraphs.
\end{definition}

\subsubsection{Typing Rules.}

The first rule we present is to check method declaration:
\begin{rules}
\entry{
  \getanncontrolling(\xabs{C.m})=[\xabs{x}_i.\xabs{m}_j\mapsto[t_j, t_j']]_{i\in I,j\in J_i}\and
  \forall i\in I,\forall j\in J_i,\, \getanncontrolled(\xabs{T}_i.\xabs{m}_j) \geq t_j \\
  \forall c\in \tacontext(\xabs{C.m}),\,\left(\begin{array}{c}
  [\xabs{x}_i.\xabs{m}_j\mapsto t_j]_{i\in I,j\in J_i},\emptyset\typep_c \xabs{s} \typed
  [\xabs{x}_i.\xabs{m}_j\mapsto t_j'']_{i\in I',j\in J_i'},\emptyset\\
  I'\subseteq I\and \forall i\in I',\,J_i'\subseteq J_i\land \forall j\in J_i,\, t_j'\leq t_j''
  \end{array}\right)
}{\xabs{C}\typep\xabs{T m(T$_1$ y$_1$, $\dots$, T$_n$ y$_n$) \{s;return e;\}}}
\end{rules}
The concluding statement of the rule means that we check the definition of method \abs{m} in class \abs{C}.
In the premise, we first check the annotation of the method.
We use two functions to access a method's annotation:
 $\getanncontrolling$ returns the information related to the {\commentfont time\_control} annotation,
 i.e., a mapping from controlled \coreid s to their pair $[t,t']$ in the annotation;
 and $\getanncontrolled$ returns the information related to the {\commentfont timed\_requires} annotation,
 i.e., how often the method must be called.
Using these two functions, the premise controls that for each controlled \coreid\ $\xabs{x}_i.\xabs{m}_j$,
 the time $t_j$ that will lapse between the beginning of the method and the first call to the \coreid\ is correct,
 i.e., it is less than the maximum time between two calls of \coreid.

The second part of the premise checks the validity of the method's body.
It collects (with $\tacontext(\xabs{C.m})$) all the execution contexts $c$ registered by the time analysis for that method,
 and analyses the validity of the method's body for each of these context individually,
 as stated in the validation predicate $\typep$ indexed with $c$.
To do that,
 it first constructs the initial contexts $\Gamma_l$ and $\Gamma_d$ where $\Gamma_l$ corresponds to the method's annotation
 (i.e., for each \coreid\ $\xabs{x}_i.\xabs{m}_j$, the time left before the next call is at most $t_j$)
 and $\Gamma_d$ is empty;
 then it analyses the method's body, which returns new contexts $\Gamma_l'$ and $\Gamma_d'$;
 and check the validity of these resulting contexts,
 i.e., the \coreid s in $\Gamma_l'$ are correct w.r.t. the annotation
 (for each \coreid\ $\xabs{x}_i.\xabs{m}_j$, the time left before the next call is at least $t_j'$ at the end of the method)
 and $\Gamma_d'$ is empty, meaning that no delegated \coreid\ need to be controlled anymore.

\begin{example}\label{ex:type:new:2}
Consider for instance the method \abs{timer} of the \abs{Controller} class in the upper part of Fig~\ref{fig:room}.
In the first part of the rule's premise, the $\getanncontrolling(\xabs{C.m})$ call gets
the method's {\commentfont time\_control} annotation which states that the method:
 takes control of the \coreid\ \abs{t.localCtrl};
 calls \coreid\ after 1 unit of time;
 and finishes its execution with \coreid\ that must be called with no delay.
Then the premise of the rule checks that this annotation is correct w.r.t. the annotation of the method \abs{Tank.localCtrl}:
 since \coreid\ will not be called after 1 unit of time, \abs{Tank.localCtrl} must not require to be called more often that this.
As \abs{Tank.localCtrl} requires to be called every 1 unit of time, this check is validated.

The second part of the premise checks that the method's body holds w.r.t.
 the method's annotation for every execution context identified by the time analysis:
 we can safely consider in this example that there is only one execution context.
Here, we initialize the $\Gamma_l$ context to the mapping $[\coreid\mapsto 1]$ and the $\Gamma_d$ context to the empty mapping.
We will see in detail with the next rules how these contexts are used and updated while type-checking statements.
But informally, we can follow the reasoning presented in Ex.~\ref{ex:type:new:1},
 to see that after executing the \abs{timer} method's body,
 the resulting context $\Gamma_l'$ is $[\coreid\mapsto 0]$,
 and the resulting context $\Gamma_d'$ is empty (as expected by the rule).
Moreover, since $\Gamma_l'$ is not empty, the last line of the premise is triggered:
 the only key in the mapping $\Gamma_l'$ is \coreid, which is declared in the method's annotation,
 and its image is 0 (i.e., the time left before the next required call to \coreid), which is consistent to the method's annotation.
\end{example}

\paragraph{Typing Statements.}

We next present the rules to handle statements within methods, based on the following judgment, which is introduced above.
\[\Gamma_l,\Gamma_d\typep_c \xabs{s}\typed \Gamma_l',\Gamma_d'\]
The first such rule deals with infinite computation.
Indeed, it is possible to delegate the control of a \coreid\ to a method that will never finish.
In that case, there is no need to take back the control of \coreid:
 we know it will be safely handled forever and we can thus simply forget about it.
\begin{rules}
\entry{
  \Gamma_l,\Gamma_d\typep_c \xabs{s}\typed \Gamma_l',\Gamma_d'\uplus[\coreid\mapsto (\mathit{fid}, t_{\min}, t_{\max}, t)] \and t_{\min} = \infty
 }{\Gamma_l,\Gamma_d\typep_c \xabs{s}\typed \Gamma_l',\Gamma_d'}
\end{rules}
This is what is written in this rule:
 if while typing a statement, we end up with a \coreid\
  that have been delegated to a method call $\mathit{fid}$ that will run forever (i.e., the minimum computation time $t_{\min}$ of $\mathit{fid}$ is infinite),
  then we can remove \coreid\ from the delegated context.

\begin{example}\label{ex:type:new:3}
This rule is suited to type-check methods similar to \abs{run} of the \abs{Mobile} class in Fig~\ref{fig:room}.
Indeed, as stated in Ex.~\ref{ex:type:new:1}, that method creates a new controller in line~\ref{ex:type:ln:newc} and
 forever delegates the control of \abs{t.localCtrl} to that controller.
Hence there is no need to keep information about that \coreid\ anymore:
 it can safely be removed from $\Gamma_d'$, which in turn makes the typing rule for method declaration hold
  as it requires for $\Gamma_d'$ to be empty.
\end{example}

\medskip

The following rule deals with conditional statements:
\begin{rules}
\entry{
  \Gamma_l,\Gamma_d\typep_c \xabs{s}_1\typed \Gamma_{l,1},\Gamma_d'\and
  \Gamma_l,\Gamma_d\typep_c \xabs{s}_2\typed \Gamma_{l,2},\Gamma_d'\\
  \dom(\Gamma_{l,1})=\dom(\Gamma_{l,2})\\
  \Gamma_l'=[\coreid\mapsto\min(\Gamma_{l,1}(\coreid), \Gamma_{l,2}(\coreid))]_{\coreid\in\dom(\Gamma_{l,1})}
}{\Gamma_l,\Gamma_d\typep_c \xabs{if(e) \{ s$_1$\} else \{ s$_2$ \}}\typed\Gamma_l',\Gamma_d'}
\end{rules}
In this rule, we type-check the two branches of the \abs{if} statement individually and require that:
 the two resulting delegated contexts are the same (i.e., they are both equal to $\Gamma_d'$);
 and that the two resulting local contexts $\Gamma_{l,1}$ and $\Gamma_{l,2}$ declare the same \coreid s.
Finally, to ensure the validity of the \abs{if} statement,
 we state that its resulting local context maps every \coreid\ to its minimum value in $\Gamma_{l,1}$ and $\Gamma_{l,2}$,
 and leave unchanged $\Gamma_d'$.

\medskip

The following rule checks sequential composition:
\begin{rules}
\entry{
  \Gamma_l,\Gamma_d\typep_c \xabs{s}_1\typed \Gamma_l',\Gamma_d'\and
  \Gamma_l',\Gamma_d'\typep_c \xabs{s}_2\typed \Gamma_l'',\Gamma_d''
}{\Gamma_l,\Gamma_d\typep_c \xabs{s$_1$ ; s$_2$}\typed\Gamma_l'',\Gamma_d''}
\end{rules}
That rules simply checks first the first statement, which gives some resulting contexts $\Gamma_l'$ and $\Gamma_d'$,
 and then the second using $\Gamma_l'$ and $\Gamma_d'$ as initial contexts.
The result of checking the sequential composition is then the result of checking the second statement.

\medskip

The \abs{await} statement is a single instruction whose only effect (of interest to our type system) is to make time pass.
This is what is written in the following rule:
 we use a specific rule, identified with the typing predicate $\typep^t_c$, to manage time passing,
 and typing the \abs{await} statement is identical to only managing the time consumed by that statement.
\begin{rules}
\entry{
  \Gamma_l,\Gamma_d\typep_c^{\mathit{t}} \xabs{await g}\typed \Gamma_l',\Gamma_d'
}{\Gamma_l,\Gamma_d\typep_c \xabs{await g}\typed\Gamma_l',\Gamma_d'}
\end{rules}
We will describe the rule for time passing using the $\typep_c^{\mathit{t}}$ judgment later.

\medskip

The following rule checks the assignment statement.
\begin{rules}
\entry{
  \Gamma_l,\Gamma_d\typep_c \xabs{rhs}\typed [\xabs{m}_i\mapsto t_i]_{i\in I},\Gamma_l',\Gamma_d'\\
  \Gamma_l''=\Gamma_l'\uplus[\xabs{x.m}_i\mapsto t_i]_{i\in I}\and
  \Gamma_l'',\Gamma_d'\typep_c^{\mathit{t}} \xabs{T x = rhs}\typed \Gamma_l''',\Gamma_d''
}{\Gamma_l,\Gamma_d\typep_c \xabs{T x = rhs}\typed\Gamma_l''',\Gamma_d''}
\end{rules}
This rule first checks the right hand side of the assignment.
Since a right hand side can create an anonymous object with methods that must be controlled,
 the result of typing a \abs{rhs} includes, in addition to the contexts $\Gamma_l'$ and $\Gamma_d'$,
 a mapping $[\xabs{m}_i\mapsto t_i]_{i\in I}$ corresponding to the newly created anonymous \coreid s.
Then the rule names these \coreid s and includes them into the local context $\Gamma_l''$ by simply adding their name \abs{x}.
Finally, the rule manages the possibility of time passing, resulting in the two final contexts $\Gamma_l'''$ and $\Gamma_d''$.

\begin{example}\label{ex:type:new:4}
In the method \abs{run} of the \abs{Mobile} class in Fig~\ref{fig:room},
 line~\ref{ex:type:ln:tank} creates a new tank \abs{t} that has the method \abs{localCtrl} that needs to be controlled.
To apply the assignment statement rule to that line,
 let first recall that the initial $\Gamma_l$ and $\Gamma_d$ contexts of that method is empty.
We will see later that the right hand side \abs{new Tank(4)} is typed
$$\emptyset,\emptyset\typep_c\xabs{new Tank(4)}\typed [\xabs{localCtrl}\mapsto 1], \emptyset,\emptyset$$
Hence, $\Gamma_l''$ in the rule is equal to $[\xabs{t.localCtrl}\mapsto 1]$.
Finally, since object creation is instantaneous, we have that $\Gamma_l'''=\Gamma_l''$ and $\Gamma_d''=\Gamma_d$:
 the rule thus correctly registers $\xabs{t.localCtrl}\mapsto 1$ to $\Gamma_l$ and keeps $\Gamma_d$ empty.
\end{example}

\medskip

The following rule, for typing the \abs{duration} statement, is identical to the one used to type the \abs{await} statement.
Indeed, the only effect of these two statements is only to make time pass,
 the only difference is that the \abs{duration} statement also blocks other processes, which is managed by the time analysis.
\begin{rules}
\entry{
  \Gamma_l,\Gamma_d\typep_c^{\mathit{t}} \xabs{duration(e)}\typed \Gamma_l',\Gamma_d'
}{\Gamma_l,\Gamma_d\typep_c \xabs{duration(e)}\typed\Gamma_l',\Gamma_d'}
\end{rules}

\medskip


The following rule deals with time passing.
\begin{rules}
\entry{
\TA(c,\xabs{s})=[t_{\min},t_{\max}]\\
\Gamma_{l,1}=[\coreid\mapsto\Gamma_l(\coreid) - t_{\max}]_{\coreid\in\dom(\Gamma_l)}\\
C=\{i\mid i\in I\land t_{i,\max}\leq t_{\min}\}\\
\Gamma_{l,2}=[\coreid_i\mapsto t_{i} + (t_{\max}-t_{i,\min})]_{i\in C}\\
\Gamma_d'=[\coreid_i\mapsto(\mathit{fid}_i, t_{i,\min}-t_{\max}, t_{i,\max}-t_{\min}, t_{i})]_{i\in I\setminus C}\\
\Gamma_l'=\Gamma_{l,1}\uplus \Gamma_{l,2}\and
\forall \coreid\in\dom(\Gamma_l'),\,\Gamma_l'(\coreid) \geq 0
}{\Gamma_l,[\coreid_i\mapsto(\mathit{fid}_i, t_{i,\min}, t_{i,\max}, t_i)]_{i\in I}\typep_c^{\mathit{t}}\xabs{s}\typed \Gamma_l',\Gamma_d'}
\end{rules}
With the function call $\TA(c,\xabs{s})$, this rule first queries the time analysis to know,
 given the current execution context $c$, what is the minimum time $t_{\min}$ and maximum time $t_{\max}$ that the current statement \abs{s} can take.
Then the rule updates the two contexts $\Gamma_l$ and $\Gamma_d$ to include this passing of time.
Updating the local context $\Gamma_l$ is quite simple,
 we just state that now the remaining time before each \coreid\ must be called has been decreased by $t_{\max}$.
Updating the delegated context $\Gamma_d$ is more subtle, since during this passing of time some method calls might have finished
 and we must retake local control of the related delegated \coreid s.
The set $C$ corresponds to all the delegated \coreid s whose control is given back to the local computation,
 since the corresponding method call is known to have finished.
We collect all these \coreid s in a new local context $\Gamma_{l,2}$,
 and for all of them, since the maximum time that lapsed between the end of the method and now is $t_{\max}-t_{i,\min}$,
 we say that the maximum amount of time that can pass until their next call is $t_{i} + (t_{\max}-t_{i,\min})$.
For the \coreid s whose control stays delegated, we simply store them in $\Gamma_d'$ and update the execution time of the related method call.

Finally, the rule checks that the locally controlled \coreid s are safe by ensuring that the time left until their next call is positive, i.e., no specified frequency is violated so far.

\begin{example}\label{ex:type:new:5:1}
Consider for instance the method \abs{timer} of the \abs{Controller} class in the upper part of Fig~\ref{fig:room}.
We already saw in Ex.~\ref{ex:type:new:2} that the initial $\Gamma_l$ context of this method is $[\coreid\mapsto 1]$
 with $\coreid=\xabs{t.localCtrl}$, while $\Gamma_d$ is empty.
The first statement of that method is \abs{await duration(1)} which clearly takes 1 unit of time:
 we can thus suppose that the time analysis states that
$$\TA(c,\xabs{await duration(1)})=[1, 1]$$
Hence, $\Gamma_{l,1}$ in the rule is $[\coreid\mapsto 0]$: \coreid\ must be called without delay.
Then, since $\Gamma_d$ is empty, so are $\Gamma_{l,2}$ and $\Gamma_d'$ in the rule.
We can then conclude the application of the rule:
 $\Gamma_l'$ is thus equal to $[\coreid\mapsto 0]$, that validates the last constraint of the rule's premise.
\end{example}

\begin{example}\label{ex:type:new:5:2}
Another interesting application of the time passing rule is the \abs{await duration(40)} statement
 line~\ref{ex:type:ln:await40} in the upper part of Fig~\ref{fig:room}.
At that point of the method's execution, the future \abs{f} is running and controlling the \coreid\ \abs{t.localCtrl}.
We can suppose that the time analysis correctly identified \abs{f}'s information,
 and so before line~\ref{ex:type:ln:await40} $\Gamma_l$ is empty,
 and $\Gamma_d$ is $[\coreid\mapsto (\xabs{f}, 41, 41, 0)]$,
 stating that \coreid\ is controlled by \abs{f},
 that \abs{f} will take exactly 41 units of time to complete,
 and that once this future completed, \coreid\ must be called right away.
Similarly to the previous example, we suppose that the time analysis for the await statement is precise:
$$\TA(c,\xabs{await duration(40)})=[40, 40]$$
Hence, the set $C$ is empty.
Following the definition of the different contexts in the rule,
 we have that $\Gamma_{l,1}$ and  $\Gamma_{l,2}$ are empty,
 and $\Gamma_d'$ is $[\coreid\mapsto (\xabs{f}, 1, 1, 0)]$.
We can then conclude the application of the rule:
 $\Gamma_l'$ is thus the empty set, and so the last constraint of the rule's premise is validated.
\end{example}

\begin{example}\label{ex:type:new:5:3}
A last interesting application of the time passing rule is the \abs{await f} statement
 line~\ref{ex:type:ln:awaitf} in Fig~\ref{fig:room}.
As we saw in the previous example, here $\Gamma_l$ is empty,
 and $\Gamma_d$ is $[\coreid\mapsto (\xabs{f}, 1, 1, 0)]$,
 stating that \coreid\ is controlled by \abs{f},
 that \abs{f} will take exactly 1 units of time to complete,
 and that once this future completed, \coreid\ must be called right away.
Similarly to the previous example, we suppose that the time analysis for the await statement is precise:
$$\TA(c,\xabs{await f?})=[1, 1]$$
Hence, the set $C$ is the rule is $\{\coreid\}$.
Following the definition of the different contexts in the rule,
 we have that $\Gamma_{l,1}$ is empty,
 $\Gamma_{l,2}$ is $[\coreid\mapsto 0]$ (the \abs{Mobile} object takes back control of \coreid\ and must call it right away),
 and $\Gamma_d'$ is empty.
We can then conclude the application of the rule:
 $\Gamma_l'$ is thus equal to $[\coreid\mapsto 0]$, that validates the last constraint of the rule's premise.
\end{example}

\paragraph{Typing Right Hand Sides and Expressions.}

The following (axiomatic) rule deals with side effect-free expression which do not have any effect by construction.
The resulting contexts are identical to the initial ones.
\begin{rules}
\axiom{\Gamma_l,\Gamma_d\typep_c \xabs{e}\typed \emptyset, \Gamma_l,\Gamma_d}
\end{rules}

\medskip

Similarly, the following rule deals with the \abs{get} right hand side, whose only effect is to pass time.
Since time passing is handled in our type system with a rule on statements (as previously described),
 no other effect is registered in this rule and the resulting contexts are identical to the initial ones.
\begin{rules}
\axiom{\Gamma_l,\Gamma_d\typep_c \xabs{e.get}\typed \emptyset, \Gamma_l,\Gamma_d}
\end{rules}
\medskip

The following rule checks object creation:
\begin{rules}
\entry{
  M = \{\xabs{m}\sht \xabs{C.m}\in\dom(\getanncontrolled)\}\and
  S=[\xabs{m}\mapsto \getanncontrolled(\xabs{C.m})]_{\xabs{m}\in M}
}{\Gamma_l,\Gamma_d\typep_c \xabs{new C(e$_1$,$\dots$,e$_n$)} \typed S, \Gamma_l,\Gamma_d}
\end{rules}
This expression does not have any effect on the contexts $\Gamma_l$ and $\Gamma_d$, but might add new \coreid s corresponding to the newly created object.
These unnamed \coreid s are stored and returned in the mapping $S$.

\medskip

The last rule of our type system deals with method call.
This rule is responsible of two main features in our type system:
 if the method call corresponds to a \coreid, we must reset the time counter in the local context $\Gamma_l$;
 and we must transfer the \coreid s delegated to this method call from $\Gamma_l$ to $\Gamma_d$.
\begin{rules}
\entry{
  \forall 1\leq i\leq n,\,\mathit{type}(\xabs{e}_i)=\xabs{T}_i \and 
  \getanncontrolling(\xabs{T$_1$.m})=[\xabs{x}_i.\xabs{m}_j\mapsto[t_j, t_j']]_{i\in I,j\in J_i}\\
  \forall i\in I,\forall j\in J_i,\, \Gamma_l[\xabs{e}_i.\xabs{m}_j] \geq t_j\\
  \Gamma_l'=\Gamma_l\setminus\{\xabs{e}_i.\xabs{m}_j \mid i\in I,j\in J_i\} \and
  \Gamma_l''=\left\{\begin{array}{ll}
    \Gamma_l'[\xabs{e}_1.\xabs{m}\mapsto\getanncontrolled(\xabs{T$_1$.m})] & \text{if $\xabs{e}_1.\xabs{m}\in\dom(\Gamma_l)$}\\
    \Gamma_l' & \text{else}
  \end{array}\right.\\
  \tacontext(c,\xabs{e$_1$!m(e$_2$,$\dots$,e$_n$)})=c' \and
  \TA(c',\xabs{T$_1$.m})=[t_{\min},t_{\max}] \and \mathit{fid}\ \text{\rm fresh}
}{\Gamma_l,\Gamma_d\typep_c \xabs{e$_1$!m(e$_2$,$\dots$,e$_n$)} \typed
  \emptyset, \Gamma_l'',\Gamma_d[\xabs{e}_i.\xabs{m}_j\mapsto (\mathit{fid}, t_{\min},t_{\max}, t_j')]_{i\in I,j\in J_i}}
\end{rules}
This rule works as follows.
First, it gets the type of every expression involved in the method call,
 and gets the information related to the {\commentfont time\_control} annotation of the callee.
Then it checks that the method call is correct w.r.t. the control annotation,
 i.e., the times given in $\Gamma_l$ for all the delegated \coreid\ are valid w.r.t. the specification of the method.
Then, it extract from $\Gamma_l$ the \coreid s that are still local in $\Gamma_l'$
 and updates $\Gamma_l'$ into $\Gamma_l''$ if the callee is a locally controlled \coreid.
And finally, it computes using the function $\tacontext$ the execution context of this method call to obtain its minimum and maximum execution time
 to generate all the relevant information in $\Gamma_d$ for the newly delegated \coreid s.

\begin{example}\label{ex:type:new:6:1}
Consider for instance the call to \abs{t.localCtrl} line~\ref{ex:type:ln:call}
 in the method \abs{timer} of the \abs{Controller} class in the upper part of Fig~\ref{fig:room}.
Let define $\coreid=\xabs{t.localCtrl}$:
 we previously discussed that before the call, $\Gamma_l$ is $[\coreid\mapsto 0]$ and $\Gamma_d$ is empty.
Then, following the first line of the premise of the rule,
 we have that $n=1$, $\xabs{T}_1=\xabs{Tank}$ and $I=\emptyset$ (since the method \abs{Tank.localCtrl} does not control anything).
The second line of the premise only checks the validity of the control delegation, which is empty in our case.
Then, following the third line of the premise,
 we have that $\Gamma_l'=\Gamma_l$ and $\Gamma_l''= [\coreid\mapsto 1]$ since \coreid\ is the method being called.
And finally, we can consider that the time analysis correctly identifies that \coreid\ does not take any time to execute
 (i.e., $\TA(c',\xabs{Tank.localCtrl})=[0,0]$);
 but since \coreid\ does not control anything, the context for delegated control stays empty after the method call.
\end{example}

\begin{example}\label{ex:type:new:6:2}
Another interesting application of the method call rule is the \abs{c.timer(t, 40)} call
 line~\ref{ex:type:ln:call40} in the upper part of Fig~\ref{fig:room}.
Let define $\coreid=\xabs{t.localCtrl}$:
 we previously discussed that before the call, $\Gamma_l$ is $[\coreid\mapsto 1]$ and $\Gamma_d$ is empty.
Then, following the first line of the premise of the rule,
 we have that $n=1$, $\xabs{T}_1=\xabs{Controller}$ and $I=\{2\}$: \abs{Controller.timer} takes control of \coreid.
The second line of the premise checks the validity of the control delegation,
 i.e., that $\Gamma_l[\coreid]$ (the maximum time allowed before calling \coreid)
  is longer than or equal to the time \abs{Controller.timer} takes to call it:
 since both numbers are 1, the check passes.
Then, following the third line of the premise,
 we have that $\Gamma_l'=\emptyset$ since \coreid\ is delegated, and so $\Gamma_l''=\emptyset$.
And finally, we can consider like in Ex.~\ref{ex:type:new:5:2} that the time analysis correctly identifies that this call to \abs{Controller.timer}
 takes exactly 41 units of time to execute (i.e., $\TA(c',\xabs{Controller.timer})=[41,41]$),
 which means that context for delegated control after the method call has the form $[\coreid\mapsto (\xabs{f}, 41, 41, 0)]$.
\end{example}

\begin{example}\label{ex:type:new:6:3}
An last interesting application of the method call rule is the \abs{c.timer(t, -1)} call
 line~\ref{ex:type:ln:newc} in the upper part of Fig~\ref{fig:room}.
Let define $\coreid=\xabs{t.localCtrl}$:
 we previously discussed that before the call, $\Gamma_l$ is $[\coreid\mapsto 0]$ and $\Gamma_d$ is empty.
Then, following the first line of the premise of the rule,
 we have that $n=1$, $\xabs{T}_1=\xabs{Controller}$ and $I=\{2\}$: \abs{Controller.timer} takes control of \coreid.
The second line of the premise checks the validity of the control delegation,
 i.e., that $\Gamma_l[\coreid]$ (the maximum time allowed before calling \coreid)
  is longer than or equal to the time \abs{Controller.timer} takes to call it:
 this check fails since $\Gamma_l[\coreid]$ is 0.
Hence, this rule correctly identifies the subtle timing error in  Fig~\ref{fig:room}.

\medskip

If instead we use the fixed version of the \abs{timer} method,
 before the call $\Gamma_l$ becomes $[\coreid\mapsto 1]$ (and $\Gamma_d$ stays empty).
Here, like in Ex.~\ref{ex:type:new:6:2}, the checks in the second line of the premise would be validated:
 our analysis would correctly validate our proposed fix.

Then, following the third line of the premise,
 we would have that $\Gamma_l'=\emptyset$ since \coreid\ is delegated, and so $\Gamma_l''=\emptyset$.
And finally, we can consider that the time analysis correctly identifies that this call to \abs{Controller.timer}
 will never finish (i.e., $\TA(c',\xabs{Controller.timer})=[\infty,\infty]$),
 which means that context for delegated control after the method call has the form $[\coreid\mapsto (\xabs{f}, \infty, \infty, 0)]$.


\end{example}


\subsection{Proof System}
For the proof system itself, we can now safely assume the specification as a post-region for the class, as long as the program is well-typed. 
\begin{theorem}[Soundness of Timed External Control]\label{thm:extern}
Let $\mathtt{C}$ be a class with a an externally timed $\mathtt{m}$ with frequency $l$.
The externally controlled timed post-region of $\mathtt{cm}$ is defined as follows:
\[\psi^{\mathsf{et}}_\xabs{C.cm} \equiv \mathtt{t} \leq l\]

Let $\mathbf{Prgm}^\Vdash$ be the set of well-typed programs according to Def.~\ref{def:type2}.
The post-region $\psi^{\mathsf{et}}_\xabs{C.cm}$ is $\mathbf{Prgm}^\Vdash$-sound for all methods \abs{m} in \abs{C}.
\end{theorem}
There can be multiple externally timed methods \EK{and it is sound to combine their corresponding post-region generators~\cite{Kamburjan21b}.}
Using the theorem, we can see that \abs{Tank} in Fig.~\ref{fig:room} is safe, if the formulas in Ex.~\ref{ex:tank} are valid and the type checker succeeds. The proof obligations are the same as the ones in Ex.~\ref{ex:timedroom}; they can be automatically closed by \texttt{KeYmaera X}~\cite{FultonMQVP15} and are available~\cite{Kamburjan21a}. 

\section{Modeling Cloud-Aware Hybrid Systems}
\label{sec:case}
Equipped with the type-and-effect system, we now investigate a bigger example to show how Hybrid Active Objects
can be used to model cyber-physical systems using a cloud infrastructure.

\paragraph{Scenario.}
We model the cloud infrastructure that is shown in Fig.~\ref{fig:iot}.
A, possibly growing, set of nodes, must be controlled by a central instance. To do so, the control instance can create new controller tasks and assign them to virtual machines (VM). Each controller task controls one node for a certain amount of time. Once it finishes controlling, a new task must be started, possibly on another VM. When a new task is created, one must pick a VM with enough free resources and, if no such VM exists, create a new one. As a VM may also run other tasks, e.g., accessing some other mechanism of the nodes, the set of VMs may grow and shrink, depending on the resource consumption.
\begin{figure}
    \centering
    \includegraphics[scale=0.6]{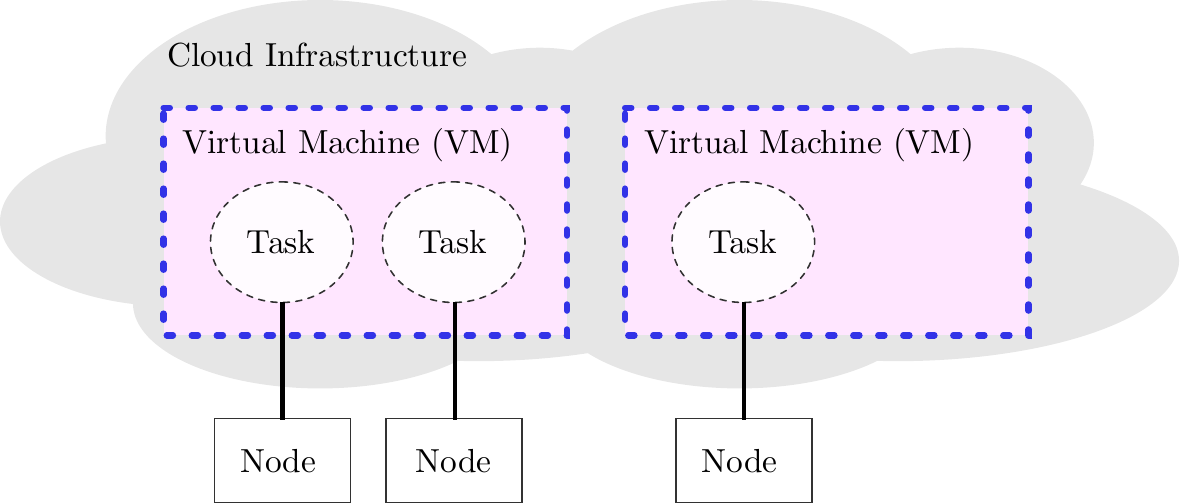}
    \caption{Cloud infrastructure to control nodes}
    \label{fig:iot}
\end{figure}

\paragraph{Resource-Modeling.}
Modeling cloud infrastructure requires the  use of \emph{resource-sensitive} actors and \emph{deployment components}~\cite{DBLP:journals/jlp/JohnsenST15} to model the cloud infrastructure.
We do not allow Hybrid Active Objects to run on deployment components, so the resource model is not relevant for our verification system. However, we consider resource-\emph{aware} HAOs, i.e., HAOs that communicate with (and are controlled by) resource-\emph{using} Timed Active Objects, which are allowed to consume resources. Communication is handled via method calls and Cooperative Contracts~\cite{DBLP:series/lncs/KamburjanDHJ20}.
A detailed introduction to resource-aware modeling with Active Objects is given by Schlatte et al.~\cite{DBLP:conf/coordination/SchlatteJKT21}, we give a short introduction into the core concepts next.

A deployment component (DC) is an abstraction over some \emph{location} that posses some \emph{resources}, that are refilled once per time unit.
Any object may run on at most one DC and each statement may consume some resource. If the DC has not enough resources to perform a resource-consuming statement, than the clock must be advanced by one time unit until the resource is refilled and can be consumed again. The process blocks the object for that time. The following code creates a DC with 3 units of the speed resource (line~\ref{line:new-DC}) and creates an object on it (line~\ref{line:obj-DC}) using the \abs{[DC: dc]} annotation. It then calls the \abs{m} method which consumes 10 resources (modeled by the \abs{[Cost: 10]} annotation. Thus, the method takes (at least) 4 time units to complete. Annotations are specific to resource modeling in \texttt{ABS}. They can be ignored \emph{for our verification}, if it can be shown that the resource model does not influence the time behavior.

\noindent\begin{abscode}
class C() { Unit m(){ [Cost: 10] skip; } }
class D() {
  Unit setup() {
    DC dc = new DeploymentComponent(map[Pair(Speed, 3)]); $\label{line:new-DC}$
    [DC : dc] C c = new C(); $\label{line:obj-DC}$
    Fut<Unit> f = c!m();
    f.get;
  }
}
\end{abscode}

\paragraph{Model.}
We now discuss our model\footnote{Available at \mbox{\url{https://formbar.raillab.de/wp-content/uploads/2021/10/nodecloud.zip}}}. It consists of three classes.
The nodes are modeled by the \abs{Tank} class in Fig.~\ref{fig:room}, that we already discussed.
The tasks are modeled by a \abs{CtrlTask} (non-hybrid) class in Fig.~\ref{fig:ex:ctrl}. The sole method of this class takes a node \abs{n} and a time \abs{until} for which it controls the node. For controlling, it calls the \abs{ctrl} method of the node once per time unit. This action consumed one resource (through the \abs{[Cost: 1]} annotation) -- it must, thus, be ensured that there are always enough resources to perform this action without delay. 

\begin{figure}[tbh]
    \centering
\begin{abscode}
class CtrlTask() {
  /*@ time_control: n.ctrl = [0,0] @*/$\label{line:ctrlTask:ann}$
    Unit ctrl(Node n, Int until){
        Rat start = timeValue(now());$\label{line:ctrlTask:line1}$
        while (timeValue(now()) <= until + start) {$\label{line:ctrlTask:line2}$
            [Cost: 1]Fut<Unit> f = n!ctrl();$\label{line:ctrlTask:line3}$
            duration(1,1);$\label{line:ctrlTask:line4}$
        }
    }
}
\end{abscode}
    \caption{A resource-aware model of a controlling task.}
    \label{fig:ex:ctrl}
\end{figure}

Finally, the managing instance is modeled as a \abs{Manager} (non-hybrid) class, shown in Fig.~\ref{fig:ex:manager1} and ~\ref{fig:ex:manager2}. 
The manager keeps a list of its DCs, corresponding to VMs, and maps each of its DCs to the number of tasks assigned to it. Periodically, here, every 10 time units, it removes all DCs which have no tasks assigned (\abs{cleanup}). It may create a new DC with the capacity for 3 tasks (\abs{createNewDc}) and it is able to return a DC with below 2 tasks if a new task arrives (\abs{getFreeDc}).
This is done in the \abs{manage} method: it takes a uncontrolled node and a time $\xabs{d} > 0$. It selects a fitting DC and creates a task on it. It then increases the number of tasks assigned to this DC and waits until the control task is finished. Once the task is finished, the number of tasks assigned to the chosen DC is reduced and after 1 time unit the procedure is repeated.

\begin{figure}[tbh]
    \centering
    \begin{abscode}[firstnumber=11]
class Manager() {
  Map<DeploymentComponent, Int> freeMap = map[];
  Unit run(){ await duration(10); this!cleanup(); this!run(); }

  DeploymentComponent createNewDc(){
    DC dc = new DeploymentComponent(map[Pair(Speed, 3)]);
    freeMap[dc] = 0; 
    return dc;}
  DeploymentComponent getFreeDc(){
    foreach (dc in keys(freeMap)){ if(freeMap[dc] <= 2) return dc; }
    return this.createNewDc();}
  /*@ time_control: n.ctrl = [0,$\color{commentcolor}\infty$] @*/
  Unit manage(Node n, Int d){/* See next figure */ }
  Unit cleanup(){ /* removes all DC with freeMap[dc] == 0	*/ }}
\end{abscode}
    \caption{A resource-aware model of a cloud infrastructure.}
    \label{fig:ex:manager1}
\end{figure}

\begin{figure}[bth]
    \centering
    \begin{abscode}[firstnumber=25]
 /*@ time_control: n.ctrl = [0,$\color{commentcolor}\infty$] @*/
Unit manage(Node n, Int d){
	DeploymentComponent target = this.getFreeDc();$\label{line:manage:1}$
	[DC : target] CtrlTask ctrlTask = new CtrlTask();$\label{line:manage:2}$
	Fut<Unit> f1 = ctrlTask!ctrl(n, d);$\label{line:manage:3}$
	freeMap[freeMap[target]]++; await f1?;$\label{line:manage:4}$
	freeMap[freeMap[target]]--;$\label{line:manage:5}$
	Fut<Unit> f2 = this!manage(n,d); await f2?; }$\label{line:manage:6}$
\end{abscode}
    \caption{A resource-aware model of a cloud infrastructure.}
    \label{fig:ex:manager2}
\end{figure}

Finally, we need a scenario to set up a number of nodes and connect them to the infrastructure. The following code sets up 10 nodes that start over 10 time units and uses each controlling task for 3 time units:

\noindent\begin{abscode}[firstnumber=35]
{ Manager manager = new Manager();$\label{line:main:1}$
  for(i in 1..10){$\label{line:main:2}$
    Node n = new Node(i); manager!manage(n, 3, i);$\label{line:main:3}$
    await duration(1);$\label{line:main:4}$
  }}  
\end{abscode}

\paragraph{Analysis}
To verify the described model, we need to perform two steps: first, we need to close all generated proof obligations, second, we need to ensure that it is typable. 
The proof obligations are already given in the previous section, so we must merely describe the typing and why the time and alias analysis succeeds. 

To check that the model is typable, we first need to identify a time analysis for it.
For simplicity, we informally describe such a $\TA$.
First, it is easy to see that the resource are correctly used in the model and do not cause time passing.
Indeed, three units of speed are declared in a deployment component,
 and each of them host the computation of at most 3 \abs{CtrlTask.ctrl} processes, each of them costing 1 unit of speed each unit of time.
Then it is easy to see that only the methods \abs{CtrlTask.ctrl}, \abs{Manager.run}, \abs{Manager.manage} and the main block take time.
Since \abs{CtrlTask.ctrl} is only called with its parameter \abs{until} equal to 3, this method always takes 3 units of time to complete
 (each iteration of its \abs{while} loop taking 1 unit of time).
The method \abs{Manager.run} runs for 10 units of time and then stops.
Since \abs{Manager.manage} synchronizes with its recursive call, it runs forever.
Finally, the main block takes 10 units of time and then concludes.

Interestingly, since every methods always has the same behavior w.r.t. time, we do not need to consider execution context for this time analysis.

\medskip

We can now apply the type system to the model, and check if a valid type statement can be derived from the rules given in Section~\ref{subsec:ts}.
Since we considered in Section~\ref{subsec:ts} that all loops have been translated away,
 we will consider here that the three loops in the models implicitly correspond to anonymous methods.
 
First, the class \abs{Node} contains only one method, \abs{ctrl} with $\getanncontrolling(\xabs{Tank.ctrl})=\emptyset$.
Hence the initial local and delegated context for typing this method are both empty.
Moreover, the two \abs{if} statements in the method do not create any controlled object,
 so the method is clearly well typed, with the type derivation of its body being:
 $$\emptyset,\emptyset\typep \xabs{if(level <= 4) drain =  1; if(level >= 9) drain = -1;}\typed \emptyset,\emptyset$$

The three other methods \abs{Manager.run}, \abs{Manager.createNewDc} and \abs{Manager.getFreeDc} have the same feature of \abs{Tank.ctrl}
 (their controlling annotation is empty, and they do not create any controlled object),
 and so they are trivially well-typed, with a type derivation similar to the one of \abs{Tank.ctrl}.

Let now consider the method \abs{CtrlTask.CtrlTask} as shown in Fig.~\ref{fig:ex:ctrl}.
Correspondingly to its annotation, the local environment used to type this method's body is $\Gamma_l=[\xabs{n.ctrl}\mapsto 0]$.
The statement of line~\ref{line:ctrlTask:line1} in Fig.~\ref{fig:ex:ctrl} does not take any time and does not have any effect on the domain of $\Gamma_l$,
 so after that statement (correspondingly to the typing rule for assignment and time passing), the local context is still equal to $\Gamma_l$.
Since the parameter \abs{until} is always equal to 3,
 we can consider for simplicity that the \abs{while} loop Line~\ref{line:ctrlTask:line2} in Fig.~\ref{fig:ex:ctrl} have been flattened into three copies of its body\footnote{As stated in Section~\ref{subsec:ts}, our type system can manage arbitrary loop statements by implicitly and automatically replacing them with dedicated method calls and synchronization. We don't use this approach in our informal explanation of the type system because even though this automatic approach always works (i.e., it is correct and complete), it adds new methods, calls and synchronization that obfuscate our explanations. Note that we do not require the recursion to be bounded.}.
Each copy starts with $\Gamma_l=[\xabs{n.ctrl}\mapsto 0]$ and $\Gamma_d=\emptyset$;
 calls \abs{n.ctrl} right away, resulting in $\Gamma_l=[\xabs{n.ctrl}\mapsto 1]$ and $\Gamma_d=\emptyset$;
 and executes \abs{duration(1,1);}, resulting in $\Gamma_l=[\xabs{n.ctrl}\mapsto 0]$ and $\Gamma_d=\emptyset$;
So, since $0\leq 0$ (the end time of \abs{n.ctrl} in the annotation is smaller or equal to the one in $\Gamma_l$),
 the concluding premises of the {\em method declaration} typing rule are validated, and so the method is well-typed.

Let now consider the method \abs{Manager.manage}.
Correspondingly to its annotation, the local environment used to type this method's body is $\Gamma_l=[\xabs{n.ctrl}\mapsto 0]$.
The statements in Lines~\ref{line:manage:1}--\ref{line:manage:3} in Fig.~\ref{fig:ex:manager2} does not take any time and does not have any effect on the domain of $\Gamma_l$,
 so after that statement (correspondingly to the typing rule for assignment and time passing), the local context is still equal to $\Gamma_l$.
Line~\ref{line:manage:3} in Fig.~\ref{fig:ex:manager2} delegates the control of \abs{n.ctrl} to the method call, which takes exactly 3 units of time,
 as stated in our discussion of $\TA$.
Hence, after Line~\ref{line:manage:3}, the local context $\Gamma_l$ is empty, and the delegated context is $\Gamma_d=[\xabs{n.ctrl}\mapsto (\mathit{fid},3,3,0)]$
 for some future name $\mathit{fid}$.
Line~\ref{line:manage:4} in Fig.~\ref{fig:ex:manager2} contains two statement.
The first one is a simple increment, and has no effect on time or on control:
 after it the local and delegated contexts are unchanged.
The second statement synchronizes with $\mathit{fid}$:
 three units of time pass, and $\mathit{fid}$ finishes, giving the control of \abs{n.ctrl} back to the local context.
Hence, since no time passed between the end of $\mathit{fid}$ and the end of the \abs{await} statement,
 that statement is well typed and results in $\Gamma_l$ being back equal to $[\xabs{n.ctrl}\mapsto 0]$ and $\Gamma_d$ being empty.
Line~\ref{line:manage:5} has not effect on time or control, so $\Gamma_l$ and $\Gamma_d$ are kept unchanged after it.
Finally, Line~\ref{line:manage:6} contains two statements.
The first statement delegates the control of \abs{n.ctrl} to another instance of the \abs{Manager.manage} method.
After this statement, we thus obtain that
 the local context $\Gamma_l$ is empty,
 and the delegated context is $\Gamma_d=[\xabs{n.ctrl}\mapsto (\mathit{fid},\infty,\infty,\infty)]$
 (recall that $\TA$ states that the computation time of \abs{Manager.manage} is infinite).
We can thus apply the {\em infinite computation} typing rule, to state that \abs{n.ctrl} is forever well manage
 and remove it from $\Gamma_d$.
Hence, before the await in Line~\ref{line:manage:6}, we get that both the local and delegated contexts are empty.
We can thus type check the await statement, that waits for an infinite amount of time without causing any control issue,
 and conclude the typing of the method.

The last instruction set we need to type check is the main block.
Similarly to the \abs{CtrlTask.CtrlTask} method, since the loop body is executed exactly 10 times, we suppose for simplicity that
 the loop has been flattened away, its body being copied 10 times in the main.
Typing the main starts with the local and delegated contexts, $\Gamma_l$ and $\Gamma_d$, being empty.
Then Line~\ref{line:main:1} has not effect on time or control, so $\Gamma_l$ and $\Gamma_d$ are kept unchanged after it.
Line~\ref{line:main:3} first creates a new \abs{Node} which has the controlled method \abs{ctrl}.
That control is given to the main, and so after this statement, $\Gamma_l=[\xabs{n.ctrl}\mapsto 1]$ ($\Gamma_d$ is still empty).
The second statement in Line~\ref{line:main:3} calls the method \abs{Manager.manage}, and thus delegates the control of \abs{n.ctrl} to it.
This method call is well typed, since the annotation of \abs{Manager.manage} stipulates that \abs{n.ctrl} will be called right away.
After this method call, we thus obtain that $\Gamma_l$ is empty,
 and $\Gamma_d$ is $[\xabs{n.ctrl}\mapsto (\mathit{fid},\infty,\infty,\infty)]$ for some future name $\mathit{fid}$.
Like during the typing of method \abs{Manager.manage}, we can apply the {\em infinite computation} typing rule
 to state that \abs{n.ctrl} is forever well manage and remove it from $\Gamma_d$.
Hence, after Line~\ref{line:main:3}, we get back an empty local context and and empty delegated context.
Line~\ref{line:main:4} makes 1 unit of time pass,  which has no effect on the main, since it has no registered controlled element.

\medskip

This concludes the typing of this model: every methods in it are well-typed, as well as the main,
 and so, following Theorem~\ref{thm:extern}, we can soundly use the specified post-region ($\mathtt{t} \leq \mathtt{l}$) in the proof obligation.

Note the modularity enforced by our system: a change in the HAO does not require us to rerun the type system, a change in the cloud system does not requires us to reprove the HAO.
\FloatBarrier
\section{Related Work}
\label{sec:related}
\paragraph{Deductive Verification of Hybrid Programming Languages}
We already discussed \ddl, which is a simple algebraic programming language. For distributed systems, \ddl has been extended to $\mathcal{Q}d\mathcal{L}$~\cite{DBLP:journals/lmcs/Platzer12b,QDL}, which is implemented in the KeYmaeraD tool~\cite{DBLP:conf/icfem/RenshawLP11}. $\mathcal{Q}d\mathcal{L}$ introduces concurrency by extending the \ddl-program variables to \emph{indexed} variables, which are manipulated using array-style statement. The concurrency model is essentially shared memory and it does not add structuring elements or special constructs to deal with concurrency.

Hybrid Rebeca~\cite{hybreb} is a language that has both constructs for discrete systems and hybrid automata, i.e., it seperates these two concepts. Its semantics is not based on classical program semantics, but on a translation into a single hybrid automaton. For verification, only model checking has been investigated. Thus, the system is neither modular nor is it able to handle unbounded systems.


Process algebras are minimalistic programming languages that have spawned several formalisms for distributed hybrid systems. None of them has been considered for type systems. The \texttt{CCPS} system~\cite{LanotteM17} is an extension of timed process algebra TPL~\cite{HennessyR95} and CCS~\cite{Milner80} uses an inbuilt notion of sensor and actuators.
The $\varphi$-calculus~\cite{RoundsS03} is an extension of the $\pi$-calculus. It has no physical processes but considers them as a part of the environment. 
The work of Khadim~\cite{Khadim} gives a detailed comparison on the process algebras \texttt{HyPA}~\cite{CuijpersR05}, \texttt{Hybrid $\chi$}~\cite{BosK03}, both extending ACP~\cite{BergstraK85}, the $\varphi$-calculus and another extension of ACP~\cite{BergstraM05}.
The \texttt{HYPE} calculus~\cite{GalpinBH13} is an approach that focuses on the composition of continuous behavior, less so on the interaction through discrete actions.

\paragraph{Compositional Deductive Verification of Hybrid Systems}
For deductive verification of hybrid models, besides \ddl, only Hybrid CSP, another process algebra, has been considered~\cite{DBLP:conf/aplas/LiuLQZZZZ10}. In its basic formulation, neither \ddl nor Hybrid CSP have structuring mechanisms for composition or modularity, and additional systems to provide a proof structure on to have been proposed. All these systems have in common that they structure the proof of a hybrid model and do not use structuring mechanisms on the language layer. Mitsch et al.~\cite{DBLP:journals/sttt/MullerMRSP18} give a methodology for composition based specific patterns used to encode components into \ddl. Bohrer and Platzer~\cite{DBLP:journals/corr/abs-2107-08852} give a proof language for Constructive Differential Game Logic, a variant of \ddl. The HHL prover~\cite{hhl} for Hybrid CSP embeds Hoare triples into Isabelle/HOL and can use its structuring mechanisms, such as lemmata. 
Baar and Staroletov~\cite{Baar} give a system to decompose \ddl proofs by transforming hybrid programs into control-flow graphs and annotating contracts to the edges. 

\paragraph{Behavioral Type Systems}
Bocchi et al.~\cite{DBLP:conf/concur/BocchiYY14} describe timed session types for a minimalistic timed process calculus with channels, based on the $\pi$-calculus. Their work checks protocol adherence and uses clock variables for time. These clocks are specified and kept track of by using linear predicates.

Majumdar et al.~\cite{MajumdarYZ20} use session types to coordinate \emph{robotic programs}. 
Robotic programs do not isolate the dynamics as hybrid objects, instead all physical processes share the same geometric space. A focus of their work is the correct partition of the geometric space and the correctness of protocols between parties in disjoint subspaces.

Avanzini and Dal Lago~\cite{DBLP:journals/pacmpl/AvanziniL17} present a type system for complexity classes of functional programs, which is at its core a time analysis. In a similar direction, there is a long line of work of cost analysis, which intersects with time analysis. For Active Objects, cost analysis has been considered by Flores-Montoya~\cite{DBLP:conf/fm/Flores-Montoya16} and Albert et al.~\cite{DBLP:journals/tocl/AlbertCJPR18}. Albert et al.~\cite{DBLP:conf/cav/AlbertGMMR20} also discuss computing $t_\mathsf{min}$. Time analysis has been used directly for deadline analysis by Laneve et al.~\cite{DBLP:journals/jlp/LaneveLPR19}.


\section{Conclusion}
\label{sec:conclusion}
We presented a verification system for distributed hybrid systems that combines deductive verification to verify object invariants with a type-and-effect system to use the global structure of the overall system.
Our system is highly modular and more expressive than prior approaches: only one proof obligation is generated per method and local changes do not require to reprove the whole system. Global analysis is performed using a lightweight type system. We can express and verify patterns with delegated control, a pattern crucial for modeling cloud-aware hybrid systems. 

This work as a further indication that hybrid programming languages are a useful modeling technique in the volatile environment of modern cyber-physical systems, and that it is possible carry over analyses techniques, such as type-and-effect systems or method-based rely-guarantee reasoning, to a hybrid setting.

Concerning the language model, we observe that an object must stay safe forever and cannot be shut down explicitly. It would be convenient, and make hybrid active objects more suitable for digital twin applications, to have some life cycle management with explicit life phases for starting, running, maintaining and shutting down a HAO. 
\section*{Acknowledgments}
This work was partially supported by the Research Council of Norway via the \texttt{SIRIUS} center (Grant Nr.~237898) and the \texttt{PeTWIN} project (Grant Nr.~294600). We thank Reiner Hähnle and Richard Bubel for extensive and constructive feedback on early drafts of this paper.

\bibliographystyle{abbrv}
\bibliography{ref}
\clearpage
\appendix

\fullpaper{
\section{Runtime Semantics}

We denote the combined local and global store with $\sigma = \rho + \tau$, which is the method composition defined by $\mathbf{dom}(\sigma) = \mathbf{dom}(\rho) \cup \mathbf{dom}(\tau)$, $\forall x \in \mathbf{dom}(\rho).~\sigma(x) = \rho(x)$
and $\forall x \in \mathbf{dom}(\tau).~\sigma(x) = \tau(x)$.

\paragraph{Evaluation of Guards and Expressions.}
To define the transition system underlying the SOS, one must be able to evaluate expressions and guards. 
Additionally to the current store, expressions and guards are evaluated with respect to a given time and continuous dynamics.
This is necessary to compute how much time may advance before a differential guard is evaluated to true. We first introduce the evaluation of expressions.

Given fixed initial values of fields $\rho$ and a local store $\tau$, the solution $f$ of an ODE is unique and $f(0) = \sigma = \rho + \tau$.
Non-physical fields do not evolve.
\begin{definition}[Evaluation of Expressions]
Let $f$ be a mapping from $\mathbb{R}^+$ to stores and $\sigma = f(0)$ the current combined state.
Let $\xabs{f}_p$ be a physical field and $\xabs{f}_d$ a non-physical field.
The semantics of fields $\xabs{f}_p$, $\xabs{f}_d$, unary operators ${\sim}$ and
binary operators
$\oplus$
after $t$ time units is defined as follows:
\begin{align*}
\sem{\xabs{f}_d}_\sigma^{f,t} &= \sigma(\xabs{f}_d) \qquad
\sem{\xabs{f}_p}_\sigma^{f,t} = f(t)(\xabs{f}_p)\\
\sem{{\sim}\xabs{e}}_\sigma^{f,t} &= {\sim}\sem{\xabs{e}}_\sigma^{f,t}\qquad\sem{\xabs{e} \oplus \xabs{e'}}_\sigma^{f,t} = \sem{\xabs{e}}_\sigma^{f,t} \oplus \sem{\xabs{e'}}_\sigma^{f,t}
\end{align*}
\end{definition}

Evaluation of guards is defined in two steps: 
First, the \emph{maximal time elapse} is computed. I.e., the maximal time that may pass without the guard expression evaluating to true
For differential guards \abs{diff e} this is the minimal time until \abs{e} becomes true.
Then, the guard evaluates to true if no time advance is needed.

\begin{definition}[Evaluation of Guards]\label{def:semant-diff-guards}
Let $f$ be a continuous dynamic of object $o$ in state $\sigma$. 
The evaluation of guards \abs{g} is defined as 
\[
\sem{\xabs{g}}_\sigma^{f,0} = \text{\normalfont\text{true}} \iff \mathit{mte}_\sigma^f(\xabs{g}) \leq 0
\]
The maximal time elapse {\normalfont($\mathit{mte}$)} is given in Fig.~\ref{fig:guard}.
\begin{figure}
\scalebox{0.85}{\begin{minipage}{\textwidth}
\begin{align*}
\mathit{mte}_\sigma^f(\xabs{duration(e)}) &= \sem{\xabs{e}}^{f,0}_{\sigma} \quad
\mathit{mte}_{\sigma}^f(\xabs{e?}) = \cased{0 &\text{ if $\sem{\xabs{e}}^{f,0}_\sigma$ is resolved} \\ \infty & \text{ otherwise}}\\
\mathit{mte}_\sigma^f(\xabs{diff e}) &= \argmin_{t\geq 0}\big(\sem{\xabs{e}}_\sigma^{f,t} = \text{\normalfont\text{true}}\big)
\end{align*}
\end{minipage}}
\caption{Computing the Maximal Time Elapse.}
\label{fig:guard}
\end{figure}
\end{definition}

The Semantics of a program run is defined as a trace, extracted from a transition sequence defined in Plotkin-style structural operational semantics~\cite{plotkin} (SOS).
At each point in time, the SOS semantic defines the state and the current global time of the program as a \emph{configuration}.
The transition system is defined by a set of conditional rewrite rules over the configurations.
Finally, the trace of a run is extracted by filling-in the gaps between the points in time where the rewrite rules generate a configuration.

To represent the state of a program, we extend the syntax of statements to \emph{runtime syntax}, which also incorporates a \abs{suspend} statement, which deschedules a process.  
The \HABS semantics are SOS semantics and are using standard SOS notation, except that (1) we use the timed extension of Bj{\o}rk et al.~\cite{BjorkBJST13}, which allows
us to keep track of time in the semantics, and (2) that time advance changes not only the time, but also the object state, because it evolves according to the \abs{physical} block.
We only present the rules relevant for the hybrid extension of \texttt{ABS} here, for the missing rules we refer to the semantics of Timed \texttt{ABS} by Bj{\o}rk et al.~\cite{BjorkBJST13}, whose semantics we conservatively extend~\cite{arxiv}: if no \abs{physical} block is used, then our semantics is the same as the one for timed \texttt{ABS}.

\begin{figure}[tbh]
\begin{align*}
\mathit{tcn} &::= \mathsf{clock}(\xabs{e})~\mathit{cn}\qquad
\mathit{cn} ::= \mathit{cn}~\mathit{cn} \synsep \mathit{fut} \synsep \mathit{msg} \synsep \mathit{ob} \qquad \mathit{fut} ::= \mathsf{fut}(\mathit{fid},\xabs{e})\\
\mathit{msg} &::= \mathsf{msg}(o, \xabs{m}, \many{\xabs{e}}, \mathit{fid})\qquad
\mathit{ob} ::= (o, \rho, {\mathit{ODE}, f}, \mathit{prc}, q)\qquad q ::= \many{\mathit{prc}}\qquad
\\
\mathit{prc} &::= (\tau,\mathit{fid},\xabs{rs})\synsep\bot\qquad
\xabs{rs} ::= \xabs{s} \synsep \xabs{suspend;s}
\end{align*}
\caption{Runtime syntax of \HABS.}
\label{fig:cfg}
\end{figure}

\begin{definition}[Runtime Syntax~\cite{BjorkBJST13}]\label{def:rsyntax}
The runtime syntax of \HABS is given by the grammar in Fig.~\ref{fig:cfg}. 
Let $\mathit{fid}$ range over future names, $o$ over object identities,
and $\rho,\tau,\sigma$ over stores (i.e., maps from fields or variables to values). 
\end{definition}

A timed configuration $\mathit{tcn}$ has a clock $\mathsf{clock}$ with
the current time as its parameter \abs{e} of \abs{Real} type and an object
configuration $\mathit{cn}$.  An object configuration 
consists of messages $\mathit{msg}$, futures $\mathit{fut}$ and
objects $\mathit{ob}$.  A message
$\mathsf{msg}(o, \xabs{m}, \many{\xabs{e}}, \mathit{fid})$ has the callee $o$, the method to be called \abs{m}, passed
parameters $\many{\xabs{e}}$ and the generated future $\mathit{fid}$ as its parameters.  A future
$\mathsf{fut}(\mathit{fid},\xabs{e})$ has as parameters the future name $\mathit{fid}$
with its return value \abs{e}.  An object
$(o, \rho, \mathit{ODE}, f, \mathit{prc}, q)$ has
an identifier $o$, an object store $\rho$, the current dynamic
$f$, an active process $\mathit{prc}$ and a set of inactive
processes $q$ as its parameters.  $\mathit{ODE}$ is taken from the class declaration.  A
process is either terminated ($\bot$) or has the form
$(\tau,\mathit{fid},\xabs{rs})$: the process store $\tau$ with the current state of
the local variables, its future identifier $\mathit{fid}$, and the statement $\xabs{rs}$
left to execute. Composition $\mathit{cn}_1~\mathit{cn}_2$
is commutative and associative.

The transition system has two parts: first, the \emph{discrete transition system}, i.e., the rules that operate on object configurations. Then, the \emph{continuous transition system} to handle time and operate on timed configurations.

\begin{figure*}[b!th]
\noindent\scalebox{0.85}{\begin{minipage}{\textwidth}
\begin{align*}
\rulename{1}\ &\big(o, \rho,{\mathit{ODE}, f}, (\tau,\mathit{fid},\xabs{await g;s}), q\big) ~\rightarrow~ \big(o, \rho,{\mathit{ODE}, f}, (\tau,\mathit{fid},\xabs{suspend;await g;s}),q\big)\\
\rulename{2}\ &\big(o, \rho,{\mathit{ODE}, f}, (\tau,\mathit{fid},\xabs{suspend;s}), q\big) ~\rightarrow~ \big(o, \rho,{\mathit{ODE}, \mathsf{sol}(\mathit{ODE},\rho)}, \bot, q\circ(\tau,\mathit{fid},\xabs{s})\big)\\
\rulename{3}\ &\big(o, \rho,{\mathit{ODE}, f}, \bot, q\circ(\tau,\mathit{fid},\xabs{await g;s})) \rightarrow~ \big(o, \rho,{\mathit{ODE}, f}, (\tau,\mathit{fid},\xabs{s}),q\big) \qquad\text{if $\sem{g}_{\rho\circ\tau}^{f,0}$ = \text{\normalfont\text{true}}}\\
\rulename{4}\ &\big(o, \rho,{\mathit{ODE}, f}, (\tau,\mathit{fid},\xabs{v = e;s}), q\big) \\&\rightarrow~ \big(o, \rho,{\mathit{ODE}, f}, (\tau[\xabs{v} \mapsto \sem{\xabs{e}}^{f,0}_{\rho\circ\tau}],\mathit{fid},\xabs{s}),q\big)\qquad
\text{if \xabs{e} contains no call or \abs{get}}\\
\rulename{5}\ &\big(o, \rho,{\mathit{ODE}, f}, (\tau,\mathit{fid},\xabs{return e;}), q\big) \rightarrow~\big(o, \rho,{\mathit{ODE}, \mathsf{sol}(\mathit{ODE},\rho)}, \bot, q\big)~\mathsf{fut}\big(\mathit{fid},\sem{\xabs{e}}^{f,0}_{\rho\circ\tau}\big)\\
\rulename{6}\ &\big(o, \rho,{\mathit{ODE}, f}, (\tau,\mathit{fid},\xabs{v = e.get;s}), q\big)~\mathsf{fut}\big(\mathit{fid},\xabs{e'}\big) 
\\&\rightarrow~ \big(o, \rho,{\mathit{ODE}, f}, (\tau,\mathit{fid},\xabs{v = e';s}), q\big) \qquad\text{if }\sem{\xabs{e}}^{f,0}_{\rho\circ\tau} = \mathit{fid}\\
\rulename{7}\ &\big(o, \rho,{\mathit{ODE}, f}, (\tau,\mathit{fid},\xabs{v = e!m(e}_1,\dots\xabs{e}_n\xabs{;s}), q\big) 
\\&\rightarrow~
\big(o,\rho,{\mathit{ODE},f},(\tau[\xabs{v}\mapsto \mathit{fid}_2],\mathit{fid},\xabs{s}),q\big) ~\mathsf{msg}\big(\sem{\xabs{e}}^{f,0}_{\rho\circ\tau},(\sem{\xabs{e}_1}^{f,0}_{\rho\circ\tau},\dots,\sem{\xabs{e}_n}^{f,0}_{\rho\circ\tau}),\mathit{fid}_2\big)\qquad\mathit{fid}_2\text{ fresh}
\end{align*}
\end{minipage}}
\caption{Selected rules for the \HABS discrete transition system.}
\label{fig:ntsem}
\end{figure*}

\paragraph{Discrete Transition System}

Fig.~\ref{fig:ntsem} gives the most important rules for the semantics of objects, the omitted rules are given in~\cite{BjorkBJST13}.
The rules match on object configurations and the matching part is rewritten accordingly.
For example, rule \rulename{5} matches on a single object and rewrites it into an object and a future. 

The rule \rulename{1} introduces a \abs{suspend} statement in front of an
\abs{await} statement.  Rule \rulename{2} consumes a \abs{suspend}
statement and moves a process into the queue of its object---at this
point, the ODEs are translated into some dynamics with $\mathsf{sol}$.   
Rule \rulename{3} activates a process
with a following \abs{await} statement, if its guard evaluates to
true. An analogous rule (not shown in Fig.~\ref{fig:ntsem}) activates
a process with any other non-\abs{await} statement.  Rule \rulename{4}
evaluates an assignment to a local variable. The rule for field is
analogous.  
Rule \rulename{5} realizes a termination (with solutions
of the ODEs) and \rulename{6} a future read. Finally, \rulename{7} is
a method call that generates a message.


\paragraph{Continuous Transition System}

For configurations, there are two rules, shown in Fig.~\ref{fig:tsem}.
Rule \rulename{i} realizes a step of an object without advancing
time. Only if \rulename{i} is not applicable, rule \rulename{ii} can
be applied. It computes the global maximal time elapse 
and advances the time in the clock and all objects.  

State change during a time advance is handled functions
$\mathit{adv}$ which are applied to all elements of a configuration
We only give two members of the family: $\mathit{adv}_{\mathit{heap}}$ takes as parameter a store $\sigma$, the dynamics $f$ and a duration $t$, 
$\mathit{adv}_{\mathit{prc}}$ takes a process $\mathit{prc}$, the dynamics $f$ and a duration $t$.
Both advance its first parameter by $t$ time units according to $f$.  
The state evolves according to the current dynamics and the guards of \abs{duration} guards and statements are decreased by $t$ (if the \abs{duration} clause is part of the first statement of an unscheduled process).
For non-hybrid Active Objects $\mathit{adv}_{\mathit{heap}}(\rho,t) = \rho$. For all other cases of \textit{adv}, the function only propagates according to the syntax in definition~\ref{def:rsyntax}.

\begin{figure}[tbh]
\scalebox{0.85}{\begin{minipage}{\columnwidth}
\begin{align*}
\rulename{i}~&\mathsf{clock}(t)~\mathit{cn}~\mathit{cn}' \rightarrow \mathsf{clock}(t)~\mathit{cn}''~\mathit{cn}' \qquad\text{ with } \mathit{cn} \rightarrow \mathit{cn}''\\
\rulename{ii}~&\mathsf{clock}(t)~\mathit{cn} \rightarrow \mathsf{clock}(t+t')~\mathit{adv}(\mathit{cn},t') \qquad\text{ if \rulename{i} is not applicable}\text{ and }\mathit{mte}(\mathit{cn}) = t'\neq \infty
\end{align*}
\end{minipage}}

\scalebox{0.85}{\begin{minipage}{\columnwidth}
\begin{align*}
\mathit{adv}_\mathit{prc}\big((\tau,\mathit{fid},\xabs{s}),f,t\big) &= (\tau,\mathit{fid},\xabs{s}) \text{ if }\xabs{s} \neq (\xabs{await duration(e); s'})\\
\mathit{adv}_\mathit{prc}\big((\tau,\mathit{fid},&\,\xabs{await duration(e);s}),f,t\big) = 
(\tau,\mathit{fid},\xabs{await duration(e+}t\xabs{);s})\\
\mathit{adv}_\mathit{heap}(\rho,f,t)(\xabs{f}) &=
\cased{
\rho(\xabs{f}) & \text{ if \abs{f} is not physical} \\
f(t)(\xabs{f})  & \text{ otherwise }}
\end{align*}
\end{minipage}}
\caption{Timed semantics of \HABS configurations.}
\label{fig:tsem}
\end{figure}

Similarly to \textit{adv}, we need to extend the \textit{mte} function to configurations as follows:

\noindent\scalebox{0.85}{\begin{minipage}{\columnwidth}
\begin{align*}
\mathit{mte}(\mathit{cn}~\mathit{cn}') &= \mathbf{min}(\mathit{mte}(\mathit{cn}),\mathit{mte}(\mathit{cn}'))\qquad\mathit{mte}(\mathit{msg}) = \mathit{mte}(\mathit{fut}) = \infty\\
\mathit{mte}(o, \rho, \mathit{ODE}, f, \mathit{prc}, q) &=\sem{\mathbf{min}_{q' \in q}(\mathit{mte}(q'),\infty)}_\rho\qquad
\mathit{mte}(\tau,\mathit{fid},\xabs{await g;s}) = \sem{\mathit{mte}(\xabs{g})}_\tau \\
\mathit{mte}(\tau,\mathit{fid},\xabs{s}) &= \infty\text{ if }\xabs{s}\neq \xabs{await g;s'}\qquad\qquad
\end{align*}
\end{minipage}} \vspace{2mm}

\noindent Note that \textit{mte} is not applied to
the currently active process, because, when \rulename{1} is not
applicable, it is currently blocking and, thus, cannot advance time.

\begin{example}~\label{ex:app}
Consider the following example, which illustrates the semantics of tank from Fig.~\ref{fig:bball}, starting at time 1:
\begin{align*}
&\mathsf{clock}(1)~\Big(o,\rho_1,\mathit{ODE},f,\bot,\big\{(\emptyset,\mathit{fid}_1,\statement^+_{\xabs{down}}),\mathit{prc}\big\}\Big)\\
\xrightarrow{(ii)}~&\mathsf{clock}(2)~\Big(o,\rho_2,\mathit{ODE},f,\bot,\big\{(\emptyset,\mathit{fid}_1,\statement^+_{\xabs{down}}),\mathit{prc}\big\}\Big)\\
\xrightarrow{(3)}~&\mathsf{clock}(2)~\Big(o,\rho_2,\mathit{ODE},f,(\emptyset,\mathit{fid}_1,\statement_{\xabs{down}}),\big\{\mathit{prc}\big\}\Big)\\
\xrightarrow{(7)}~&\mathsf{clock}(2)~\Big(o,\rho_2,\mathit{ODE},f,(\emptyset,\mathit{fid}_1,\dots),\big\{\mathit{prc}\big\}\Big)\quad\mathsf{msg}(o,\mathit{log},\mathit{fid}_2)
\end{align*}
The store is $\rho_1 = \{\xabs{level} \mapsto 4,\xabs{drain} \mapsto -1, \xabs{log} \mapsto \mathit{l}\}$ and two processes are suspended.
The one for \xabs{up} is denoted $\mathit{prc}$, the one for \xabs{down} has the remaining statement $\statement^+_{\xabs{down}}$, which is the whole method body.

Nothing can execute without advancing time, so time is advanced by 1 time unit (using rule (ii)) until the store changes to 
$\rho_1 = \{\xabs{level} \mapsto 3,\xabs{drain} \mapsto -1,\xabs{log} \mapsto \mathit{l}\}$.

This enables rule (3) to schedule the process for \xabs{down}, where the \abs{await} statement is removed: $\statement_{\xabs{down}}$ is the method body without the leading suspension.
Finally, rule (7) is used to generate a message to call the other object.
\end{example}
\section{Proofs}

\subsection{Proof for Theorem~\ref{thm:extern}}

To prove this Theorem, we first formalize the Time Analysis Definition we introduced in Section~\ref{subsec:ts} and define when such analysis is correct.

\subsubsection{Time Analysis}

Our formalization of the time analysis is done in two parts.
First we precisely define the properties of the expressions in $\CF$:
\begin{definition}\label{def:counting-framework}
A {\em Counting Framework} is a tuple $\CF=(\cfexps,\cfminus,\cfisinfty,\cfispositive)$ where:
 $\cfexps$ is the set of all counting expressions;
 $\cfminus:(\cfexps\times\cfexps)\rightarrow\cfexps$ is the minus operation;
 $\cfisinfty\subseteq\cfexps$ is a predicate over counting expressions stating which ones are known to correspond to infinity;
 and $\cfispositive\subseteq\cfexps$ is a predicate over counting expressions stating which ones are known to be positive or null.

We say that $\CF=(\cfexps,\cfminus,\cfisinfty,\cfispositive)$ is {\em valid}
 iff $\mathbb{Q}\subseteq\cfexps$ and the operations $\cfminus$, $\cfisinfty$ and $\cfispositive$ behave as expected with rational as parameters.
\end{definition}
It is important to note that every expression manipulation done in our typing rules in Section~\ref{subsec:ts}
 can be expressed with the reduced set of operations given in Definition~\ref{def:counting-framework}.

\medskip

Using the notations
 $\allmethodnames$ for all the possible qualified method names \abs{C.m},
 $\allstatements$ for all the possible statements in an ABS AST,
 and $\allrhs$ for all the possible \abs{rhs} in an ABS AST,
 we can now formally define a time analysis:
\begin{definition}
A {\em Time Analysis} (TA) is a tuple $\TA=(\CF,\tapreds,\tacontext,\tamap)$ where:
 $\CF=(\cfexps,\cfminus,\cfisinfty,\cfispositive)$ is a valid counting framework;
 $\tapreds$ is a set of {\em execution contexts};
 $\tacontext=\tacontext_1\cup\tacontext_2$ with 
   $\tacontext_1:\allmethodnames\mapsto2^\tapreds$ maps all methods in the program to a set of execution contexts
   and $\tacontext_1:(\tapreds\times\allrhs)\mapsto\tapreds$ maps all method calls in the program, given a caller execution context, to the callee execution context;
 and $\tamap:((\tapreds\times\allmethodnames)\cup(\tapreds\times\allstatements))\mapsto(\cfexps\times\cfexps)$ is a partial function giving
  for pairs of execution contexts and qualified method names, and for pairs of execution contexts and statements, their minimum and maximum execution time.
\end{definition}
Note that in our typing rules, we identified $\CF$ with $\cfexps$ and $\TA$ with $\tamap$.

\medskip

In the previous definition, nothing ensures that the information stated in the time analysis $\TA$ does correspond to the actual execution of a program.
This notion of {\em validity w.r.t. a program \abs{Prgm}} is described in the following.
More precisely, for a $\TA$ to be valid w.r.t. a program, it must be valid for all its possible executions, i.e., for all its runs.
In the following, we introduce some notations for the different manipulation we do on runs to define the validity of a time analysis.
\begin{notation}
Suppose given a (possibly infinite) run $R=\mathit{tcn}_0\rightarrow\mathit{tcn}_1\rightarrow\dots$.
For $i\in\mathbb{N}$, we first write $R_{=i}=\mathit{tcn}_i$ and $R_{\geq i}=\mathit{tcn}_i\rightarrow\dots$.
We write $\mathit{cn}_{|t}\in R$ iff there exist $i$ and $\mathit{cn}_i$ such that $R_{=i}=\mathsf{clock}(t)\  \mathit{cn}\ \mathit{cn}_i$.
We also write $\mathit{prc}_{|t}\in R$ iff there exists $o$, $\rho$, ${\mathit{ODE}, f}$ and $q$ such that $(o, \rho, {\mathit{ODE}, f}, \mathit{prc}, q)_{|t}\in R$.
Finally, we write $\mathit{cn}_{|t}\rightarrow^*\mathit{cn}_{|t'}'\in R$ (resp. $\mathit{prc}_{|t}\rightarrow^*\mathit{prc}_{|t'}'\in R$)
 iff there exists $i\leq j$ such that $\mathit{cn}_{|t}\in R_{\geq i}$ and $\mathit{cn}_{|t'}'\in R_{\geq j}$
 (resp. $\mathit{prc}_{|t}\in R_{\geq i}$ and $\mathit{prc}_{|t'}'\in R_{\geq j}$).
%
%
%
\end{notation}
The following definition states when a time analysis $\TA$ is valid w.r.t. a program.
This definition is structured in four parts:
 first we introduce {\em timed instances}, that link $\TA$ to the runtime states of the program;
 then we define {\em context-validity}, that ensures that the $\tacontext$ function is well constructed;
 then we define {\em method-validity}, that ensures that the time information for methods is correct;
 and finally we define {\em statement-validity}, that ensures that the time information for statements is correct.
\begin{definition}\label{def:ta-validity}
Given a type analysis $\TA=(\CF,\tapreds,\tacontext,\tamap)$,
 a {\em timed instance} $\TI$ is a partial function from future names to execution contexts $c$ in $\tapreds$.

We say that $\TA$ is {\em context-valid} for a run $R$ and a timed instance $\TI$ iff
 for all time $t$ and all
$$\begin{array}{c}
\left(\begin{array}{c}
  \big(o, \rho,{\mathit{ODE}, f}, (\tau,\mathit{fid},\xabs{T v = e!m(e$_1$,$\dots$,e$_n$);s}), q\big)_{|t} \\
  \rightarrow^*\\
  \left(\big(o,\rho,{\mathit{ODE},f},(\tau[\xabs{v}\mapsto \mathit{fid}_2],\mathit{fid},\xabs{s}),q\big)
  ~\mathsf{msg}\big(\xabs{e}',\xabs{e}_1',\dots,\xabs{e}_n'),\mathit{fid}_2\big)\right)_{|t}
\end{array}\right)\in R
\end{array}$$
then $\mathit{fid}\in\dom(\TI)$, $\mathit{fid}_2\in\dom(\TI)$, and $\tacontext(\TI(\mathit{fid}), \xabs{v = e!m(e$_1$,$\dots$,e$_n$;s)}=\TI(\mathit{fid}_2)$.
 
We say that $\TA$ is {\em method-valid} for a run $R$ and a timed instance $\TI$ iff
 for all $(\mathsf{msg}(o, \xabs{m}, \many{\xabs{e}}, \mathit{fid}))_{|t}\in R$, with $\mathit{o}=\xabs{C}$, we have that $\mathit{fid}\in\dom(\TI)$
 and, writing $\tamap(\TI(\mathit{fid}), \xabs{C.m})=[t_{\min}, t_{\max}]$, both:
\begin{enumerate}
\item there exists $t''\geq t_{\max}$ such that $(\mathsf{fut}(\mathit{fid},\xabs{e}))_{|t+t''}\in R$
\item for all $t'\leq t_{\min}$, there exists $t''\geq t'$ such that $(\tau,\mathit{fid},\xabs{rs})_{|t+t''}\in R$
\end{enumerate}

We say that $\TA$ is {\em statement-valid} for a run $R$ and a timed instance $\TI$ iff
 for all $(\tau,\mathit{fid},\xabs{s;rs})_{t}\rightarrow^*(\tau',\mathit{fid},\xabs{rs})_{t'}\in R$, we have that $\mathit{fid}\in\dom(\TI)$
 and, writing $\tamap(\TI(\mathit{fid}), \xabs{s})=[t_{\min}, t_{\max}]$, we also have $t_{\min}\leq t'-t\leq t_{\max}$.

Finally, we say that $\TA$ is {\em valid} for a program \abs{Prgm} iff for all runs $R$ of this program, there exists a timed instance $\TI$ such that
 $\TA$ is context-, method- and statement-valid for $R$ and $\TI$.
\end{definition}

\subsubsection{Run Annotations}

To be able to prove the correction of our type system, we must somehow apply it on every possible run of a program,
 and show that it implies that not to much time passes between two calls of a controlled method.
The difficulty here is that runs can contain an infinite number of steps before time advances,
 and so a classic inductive reasoning (like subject reduction~\cite{subject-reduction}) cannot be applied.

We solve this issue using a co-inductive~\cite{coinduction} approach to apply the type system onto a run.
Our application of the type system is done in three steps.
The first step is implemented in the following definition, and is the one dealing with co-induction.
It introduce a {\em control flow} $\Delta$ that extracts from each state in a run $R$ which future controls which \coreid.
More precisely, since a \coreid\ at runtime is a pair $o.\xabs{m}$, it also stores the variable \abs{x} used in the type derivation,
 to reconstruct the corresponding \coreid \abs{x.m} used at static time.
And to manage the fact that when a method ends, the control of its parameter is given back to the caller,
 $\Delta$ not only stores for every state of $R$ the current controller, but also the callers, in a stack structure.

The correct construction of $\Delta$ is ensured by checking that each statement that changes control is correctly manipulated by it.
Note that since $\Delta$ manipulates stacks, we use the standard stack operations in the definition:
 $\mathit{push}(s,e)$ returns the stack $s$ extended with the element $e$;
 $\mathit{last}(s)$ returns the last element of the stack;
 and $\mathit{pop}(s)$ returns the stack, minus its last element
\begin{definition}
Given a run $R=\mathit{tcn}_0\rightarrow\mathit{tcn}_1\rightarrow\dots$, the {\em control flow} $\Delta=\delta_0\rightarrow\delta_1\rightarrow\dots$ of $R$
gives for each state $\mathit{tcn}_i$ of $R$ a corresponding {\em control information} $\delta_i$ that maps \coreid s $o.\xabs{m}$ to stacks of future names,
 such that for all $\mathit{tcn}_i\rightarrow\mathit{tcn}_j\in R$, we have either:
\begin{itemize}
\item $\delta_j=\delta_i[o.\xabs{m}\mapsto[(\mathit{fid}, \xabs{v}]]_{\xabs{m}\in\{\xabs{m}\mid \xabs{C.m}\in\dom(\getanncontrolled)\}}$ if
\begin{multline*}
\mathit{tcn}_i\rightarrow\mathit{tcn}_j=\\
  \big(o', \rho',{\mathit{ODE}', f'}, (\tau,\mathit{fid},\xabs{T v = new C(e$_1$,$\dots$,e$_n$);s}), q'\big) \\
  \rightarrow
  \left(\big(o',\rho',{\mathit{ODE}',f'},(\tau[\xabs{v}\mapsto o],\mathit{fid},\xabs{s}),q'\big)
  ~(\big(o,\rho'',{\mathit{ODE}'',f''}, \bot,\emptyset\big)\right)
\end{multline*}
\item $\delta_j=\delta_i[o_i.\xabs{m}_j\mapsto\mathit{push}(\delta_i(o_i.\xabs{m}_j), (\mathit{fid_2}, \xabs{x}_i))]_{i\in I,j\in J_i}$ if
\begin{multline*}
\mathit{tcn}_i\rightarrow\mathit{tcn}_j=\\
  \big(o, \rho,{\mathit{ODE}, f}, (\tau,\mathit{fid},\xabs{T v = e!m(e$_1$,$\dots$,e$_n$);s}), q\big) \\
  \rightarrow
  \left(\big(o,\rho,{\mathit{ODE},f},(\tau[\xabs{v}\mapsto \mathit{fid}_2],\mathit{fid},\xabs{s}),q\big)
  ~\mathsf{msg}\big(\xabs{e}',\xabs{e}_1',\dots,\xabs{e}_n'),\mathit{fid}_2\big)\right)
\end{multline*}
 with $\mathit{type}(o)=\xabs{C}$ and
  $\getanncontrolling(\xabs{C.m})=[\xabs{x}_i.\xabs{m}_j\mapsto[t_j, t_j']]_{i\in I,j\in J_i}$
\item $\delta_j=\delta_i[o_i.\xabs{m}_j\mapsto\mathit{pop}(\delta_i(o_i.\xabs{m}_j))]_{\mathit{last}(\delta_i[o_i.\xabs{m}_j])=\mathit{fid}}$ if
\begin{multline*}
\mathit{tcn}_i\rightarrow\mathit{tcn}_j=\\
\big(o, \rho,{\mathit{ODE}, f}, (\tau,\mathit{fid},\xabs{return e;}), q\big) \\\rightarrow \big(o, \rho,{\mathit{ODE}, \mathsf{sol}(\mathit{ODE},\rho)}, \bot, q\big)~\mathsf{fut}\big(\mathit{fid},\xabs{e}'\big)
\end{multline*}
\item $\delta_j=\delta_i$ otherwise
\end{itemize}
\end{definition}

The second step in our application is to retrieve information about method calls,
 i.e., to which method corresponds a future.
This is done in the following definition.
\begin{definition}
We say that a map $\futuremethod$ from future names to qualified method names is {\em correct} w.r.t. a program run $R=\mathit{tcn}_0\rightarrow\dots$ iff
 we have $\mathit{fid}\in\dom(\futuremethod)$ and $\futuremethod(\mathit{fid})=\xabs{C.m}$ for all
 $$\begin{array}{c}
\left(\begin{array}{c}
  \big(o, \rho,{\mathit{ODE}, f}, (\tau,\mathit{fid},\xabs{T v = e!m(e$_1$,$\dots$,e$_n$);s}), q\big)_{|t} \\
  \rightarrow^*\\
  \left(\big(o,\rho,{\mathit{ODE},f},(\tau[\xabs{v}\mapsto \mathit{fid}_2],\mathit{fid},\xabs{s}),q\big)
  ~\mathsf{msg}\big(\xabs{e}',\xabs{e}_1',\dots,\xabs{e}_n'),\mathit{fid}_2\big)\right)_{|t}
\end{array}\right)\in R
\end{array}$$
with $\mathit{type}(\xabs{e})=\xabs{C}$.
\end{definition}

Finally, we can combine all the elements we constructed to apply our type system onto a run,
 i.e., giving to each execution step in a run $R$ a mapping $\Gamma$ from \coreid s to how much time can pass until their next call.
\begin{definition}\label{def:ts-runtime}
Suppose given a well-typed program \textsf{Prgm}, a valid time analysis for it $\TA$.
Consider moreover a run $R=\mathit{tcn}_0\rightarrow\dots$ of this program, with corresponding $\TI$, $\Delta$ and correct $\futuremethod$.
Then, for every $\delta_i\in\Delta$ and every $\text{\coreid}\in\dom(\delta_i)$, we define $\Gamma_i(\text{\coreid})=\Gamma_l(\xabs{x})$ where
 $\mathit{last}(\delta_i(\text{\coreid}))=(\mathit{fid},\xabs{x})$,
 $(\tau,\mathit{fid},\xabs{s})\in\mathit{tcn}_i$,
 $\TI(\mathit{fid})=c$, $\futuremethod(\mathit{fid})=\xabs{C.m}$,
 and $\Gamma_l$ is the local context in the type derivation of $\Gamma_l,\Gamma_d\typep_c \xabs{s}\typed \Gamma_l',\Gamma_d'$ while typing \abs{C.m}.
\end{definition}

\subsubsection{Correction of the Type System}

We now have all the elements to prove that all the time-related annotation in a program are validated if that program type checks.
The proof is structured in two parts.
First, we prove that if a controlled object is ever created during a run of that program, then that run must last for an infinite amount of time
 (which implies that the run has an infinite number of steps).
\begin{lemma}\label{lem:infty-computation}
Consider a program \textsf{Prgm}, well-typed for Definition~\ref{def:type2}, and a run $R=\mathit{tcn}_0\rightarrow\dots$ of that program.
Suppose moreover that there exist \abs{C.m} and $i$ such that $\xabs{C.m}\in\dom(\getanncontrolled)$
 and $\mathit{tcn}_i=(\mathit{tcn}'\ (o, \rho, {\mathit{ODE}, f}, \mathit{prc}, q))$ with $\mathit{type}(o)=\xabs{C}$.
Then for all $t\in\mathbb{Q}$, there exists $j$, $\mathit{cn}$ and $t'\geq t$ such that $\mathit{tcn}_j=\mathsf{clock}(t')~\mathit{cn}$.
\end{lemma}
\begin{proof}
Since $\mathit{tcn}_i=(\mathit{tcn}'\ (o, \rho, {\mathit{ODE}, f}, \mathit{prc}, q))$, there must exist 
$$\begin{array}{c}
\left(\begin{array}{c}
  \big(o', \rho',{\mathit{ODE}', f'}, (\tau,\mathit{fid},\xabs{T v = new C(e$_1$,$\dots$,e$_n$);s}), q'\big)_{|t} \\
  \rightarrow^*\\
  \left(\big(o',\rho',{\mathit{ODE}',f'},(\tau[\xabs{v}\mapsto o],\mathit{fid},\xabs{s}),q'\big)
  ~(\big(o,\rho'',{\mathit{ODE}'',f''}, \bot,\emptyset\big)\right)_{|t}
\end{array}\right)\in R
\end{array}$$
where the future $\mathit{fid}$ corresponds to a typed-checked method $\xabs{C}'\xabs{.m}'$.
Since that method is typed-checked, \abs{v} is added to its local context $\Gamma_l$ when created,
 but nor the local context $\Gamma_l'$ or the delegated context $\Gamma_d'$ resulting of typing $\xabs{C}'\xabs{.m}'$ can contain \abs{v}.
Hence, it is clear by the construction of the typing rules, that for all $c\in\tacontext(\xabs{C$'$.m$'$})$,
 $\xabs{C}'\xabs{.m}'$ contains, after the creation of \abs{v}, a call $\xabs{T w = e!m(e$_1$,$\dots$,e$_n$)})$,
 such that one of the \abs{e$_i$} is equal to \abs{v} and $t_{\min}=\infty$ with $\TA(\tacontext(c,\xabs{e!m(e$_1$,$\dots$,e$_n$)}), \xabs{C.m})=[t_{\min}, t_{\max}]$.

By Definition~\ref{def:ta-validity} (and in particular, the definition of method-validity), this means that 
 for all $t'$, there exists $t''\geq t'$ such that $(\tau,\mathit{fid},\xabs{rs})_{|t+t''}\in R$ where $\mathit{fid}'$ is the future id of the method call
 $\xabs{e!m(e$_1$,$\dots$,e$_n$)})$.
Consequently, for all $t$, there exists $t'\geq t$ and $i$ such that $R_{=i}=\mathsf{clock}(t')\ \mathit{cn}$:
 there is no upper bound for the time in $R$.
\end{proof}

The second part of the proof uses Definition~\ref{def:ts-runtime} and the property stated in the previous Lemma to prove the correction of our type system
\begin{theorem}\label{typetheorem}
Given a well-typed program \textsf{Prgm}, then
 for every run $R=\mathit{tcn}_0\rightarrow\dots$ of that program,
 for every $(o, \rho, {\mathit{ODE}, f}, \mathit{prc}, q)_{|t}\in R$ with $\mathit{type}(o)=\xabs{C}$ and there exists $\xabs{C.m}\in\dom(\getanncontrolled)$,
 there exists $t'\leq t+\getanncontrolled(\xabs{C.m})$, some  $\many{\xabs{e}}$ and $\mathit{fid}$ with $(\mathsf{msg}(o, \xabs{m}, \many{\xabs{e}}, \mathit{fid}))_{|t'}\in R$.
\end{theorem}
\begin{proof}
By definition of $(o, \rho, {\mathit{ODE}, f}, \mathit{prc}, q)_{|t}\in R$, there must exist $i$ and $\mathit{cn}$ such that
 $\mathit{tcn}_i=\mathsf{clock}(t)\ \mathit{cn}\ (o, \rho, {\mathit{ODE}, f}, \mathit{prc}, q)$.
By Lemma~\ref{lem:infty-computation}, there exists $i'$, $t''\geq t+\getanncontrolled(\xabs{C.m})$ and $\mathit{cn}'$ such that $\mathit{tcn}_{i'}=\mathsf{clock}(t'')\ \mathit{cn}'$.
It is clear that $i<i'$.

Let now consider the sequence $S=\Gamma_0\rightarrow\dots$ constructed as in Definition~\ref{def:ts-runtime}.
Note that by construction of $S$ and the {\em time passing} typing rules, for all $\Gamma_k\in S$, the image of $\Gamma_k$ only contains positive numbers,
 and for all $o'.\xabs{m}\in\dom(\Gamma_k)$, we have $\Gamma_k(o'.\xabs{m})\leq \getanncontrolled(\xabs{C$'$.m})$ where $\mathit{type}(o')=\xabs{C}'$.
Moreover, since the time analysis $\TA$ used to construct $S$ is valid w.r.t. \textsf{Prgm},
 by construction of the {\em time passing} typing rules, and the fact that only the {\em method call} typing rule can increase the value of $\Gamma_k(o.\xabs{m})$,
 for every $i\leq j<k$, with $\mathit{tcn}_j=\mathsf{clock}(t_j)\ \mathit{cn}_j$, $\mathit{tcn}_k=\mathsf{clock}(t_k)\ \mathit{cn}_k$ and such that there is no
 $j\leq l\leq k$ such that $(\mathsf{msg}(o, \xabs{m}, \many{\xabs{e}}, \mathit{fid}))\in \mathit{tcn}_l$ for some $\many{\xabs{e}}$ and $\mathit{fid}$,
 we have that $\Gamma_k(o.\xabs{m})\leq \Gamma_j(o.\xabs{m}) - (t_k-t_j)$.
 
Applying this property to $i$ and $i'$, we obtain that $\Gamma_{i'}(o.\xabs{m})\leq \Gamma_i(o.\xabs{m}) - (t_{i'}-t_i)$.
Since $\Gamma_i(o.\xabs{m})\leq\getanncontrolled(\xabs{C.m})$, if there is no call to $o.\xabs{m}$ in $R$ between $\mathit{tcn}_i$ and $\mathit{tcn}_{i'}$,
 then $\Gamma_k(o.\xabs{m})<0$, which is impossible.

We can thus conclude that for every state $\mathit{tcn}_{i'}=\mathsf{clock}(t'')\ \mathit{cn}'$ in $R$ such that $t''\geq t+\getanncontrolled(\xabs{C.m})$,
 there exists $i \leq i''< i'$ such that $(\mathsf{msg}(o, \xabs{m}, \many{\xabs{e}}, \mathit{fid}))\in \mathit{tcn}_{i''}$.
Hence, there exist such a $i''$ such that $\mathit{tcn}_{i''}=\mathsf{clock}(t')\ \mathit{cn}'$ with $t'\leq t+\getanncontrolled(\xabs{C.m})$,
 which gives us the result.
\end{proof}

Theorem~\ref{thm:extern} is now a simple corollary of Theorem~\ref{typetheorem}: as for every moment in the run the time to the next call on a specified method is bounded by its specified frequency $\mathtt{t'}$, the time between two method starts can be at most $\mathtt{t'}$, as well. Thus, each suspension subtrace is bounded in time by $\mathtt{t'}$, which is exactly the condition expressed by the theorem.

}
\end{document}